\documentclass[conference]{sig-alternate}

\usepackage{algorithm}
\usepackage{algorithmic}
\usepackage{amssymb}

\usepackage{amsmath}
\usepackage{amsfonts}
\usepackage{amssymb}
\usepackage{url}
\usepackage{cite}
\usepackage{epsfig}
\usepackage{xspace}
\usepackage{graphicx}
\usepackage{subfigure}
\usepackage{wrapfig}
\usepackage{color}
\usepackage{listings}
\usepackage{stmaryrd}
\usepackage{xspace}

\usepackage{algorithm}
\usepackage{algorithmic}
\usepackage{multicol}
\usepackage{graphicx}
\usepackage{color}

\newcommand{\algorithmicinput}{\textbf{Input:}}

\newcommand{\INPUT}{\item[\algorithmicinput]}

\newcommand{\algorithmicoutput}{\textbf{Output:}}

\newcommand{\OUTPUT}{\item[\algorithmicoutput]}

\newtheorem{theorem}{Theorem}
\newtheorem{corollary}{Corollary}
\newtheorem{observation}{Observation}
\newtheorem{lemma}{Lemma}
\newtheorem{definition}{Definition}
\newtheorem{assumption}{Assumption}
\newtheorem{remark}{Remark}

\renewcommand{\baselinestretch}{0.95}

\begin{document}
\bibliographystyle{plain}
\title{Stabilization and Fault-Tolerance in Presence of Unchangeable Environment Actions}

\numberofauthors{4}
\author{
\alignauthor
Mohammad Roohitavaf\\
       \affaddr{Computer Science and Engineering Department}\\
       \affaddr{Michigan State  University}\\
       \affaddr{East Lansing, Michigan 48824, USA}\\
       \email{roohitav@cse.msu.edu}
\alignauthor
Sandeep Kulkarni\\
       \affaddr{Computer Science and Engineering Department}\\
       \affaddr{Michigan State  University}\\
       \affaddr{East Lansing, Michigan 48824, USA}\\
       \email{sandeep@cse.msu.edu}
}

\maketitle

\begin{abstract}
We focus on the problem of adding fault-tolerance to an existing concurrent protocol in the presence of {\em unchangeable environment actions}. Such unchangeable actions occur in practice due to several reasons. One instance includes the case where only a subset of the components/processes can be revised and other components/processes must be as is. Another instance includes cyber-physical systems where revising physical components may be undesirable or impossible. These actions differ from faults in that they are simultaneously {\em assistive} and {\em disruptive}, whereas faults are only disruptive.  For example, if these actions are a part of a physical component, their execution is essential for the normal operation of the system. However, they can potentially disrupt actions taken by other components for dealing with faults. Also, one can typically assume that fault actions will stop for a long enough time for the program to make progress. Such an assumption is impossible in this context.

We present algorithms for adding stabilizing fault-tolerance, failsafe fault-tolerance and masking fault-tolerance. Interestingly, we observe that the previous approaches for adding stabilizing fault-tolerance and masking fault-tolerance cannot be easily extended in this context. However, we find that the overall complexity of adding these levels of fault-tolerance remains in P (in the state space of the program). We also demonstrate that our algorithms are sound and complete.

\end{abstract}


\section{Introduction}
\label{sec:intro}

In this paper, we focus on the problem of model repair for the purpose of making the model stabilizing or fault-tolerant. Model repair is the problem of revising an existing model/program so that it satisfies new properties while preserving existing properties. It is desirable in several contexts such as when an existing program needs to be deployed in a new setting or to repair bugs. Model repair for fault-tolerance enables one to separate the fault-tolerance and functionality so that the designer can focus on the functionality of the program and utilize automated techniques for adding fault-tolerance. It can also be used to add fault-tolerance to a newly discovered fault.

This paper focuses on performing such repair when some actions cannot be removed from the model. We refer to such transitions as {\em unchangeable environment actions}. There are several possible reasons that actions can be unchangeable. Examples include scenarios where the system consists of several components --some of which are developed in house and can be repaired and some of which are third-party and cannot be changed. They are also useful in systems such as Cyber-Physical Systems (CPSs) where modifying physical components may be very expensive or even impossible.

The environment actions differ from fault actions considered in \cite{bka12}. Fault actions are assumed to be temporary in nature, and all the previously proposed algorithms to add fault-tolerance in \cite{bka12}, work only with this important assumption that faults finally stop occurring. However, unlike fault actions, environment actions can keep occurring. Environment actions also differ from adversary actions considered in \cite{bk11sss} or in the context of security intrusions. In particular, the adversary intends to cause harm to the system. By contrast, environment actions can be collaborative as well.
In other words, the environment actions are simultaneously collaborative and disruptive. The goal of this work is to identify whether it is possible for the program to be repaired so that it can utilize the assistance provided by them while overcoming their disruption. To give an intuition of the role of the environment and the difference between program, environment, and fault actions, next, we present the following example. 

{\bf An intuitive example to illustrate the role of environment.  }  This intuitive example is motivated by a simple pressure cooker (see Figure \ref{preCookerFig}). The environment (heat source) causes the pressure to increase. In the subsequent discussion, we analyze this pressure cooker when the heat source is always on. 
There are two mechanisms to decrease the pressure, a vent and an overpressure valve. For sake of presentation, assume that pressure is below 4 in normal states. If the pressure increases to {4 or 5}, the vent mechanism reduces the pressure by 1 in each step. However, the vent may fail (e.g., if something gets stuck at the vent pipe), and its pressure reduction mechanism becomes disabled. If the pressure reaches 6, the overpressure valve mechanism causes the valve to open resulting in an immediate drop in pressure to be less than 4. We denote the state where pressure is $a$ by $s_a$ when the vent is working, and by state $fs_a$ when the vent has failed.

\begin{figure}
\begin{center}
\includegraphics[width=80mm,scale=0.5]{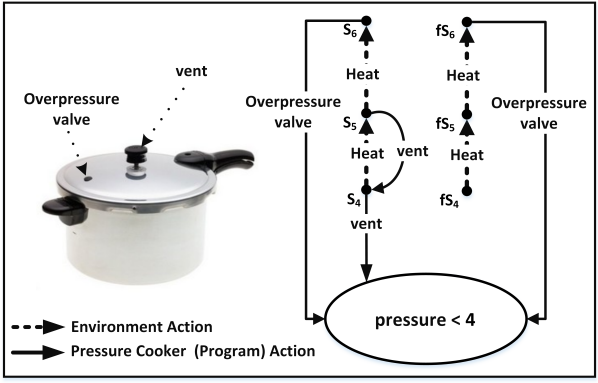}
\caption{An intuitive example to illustrate the role of environment actions. For sake of readability, fault actions (e.g. actions to $fs_4$) are removed from the diagram.}
\label{preCookerFig}
\end{center}
\end{figure}

Our goal in the subsequent discussion is to model the pressure cooker as a program and identify an approach for the role of the environment and its interaction with the program so that we can conclude this requirement: \textit{starting} from \textit{any} state identified above, the system reaches a state where the pressure is less than 4. 

Next, we argue that the role of the environment differs from that of fault actions and program actions. In turn, this prevents us from using existing approaches such as \cite{bka12}. Specifically,

\begin{itemize}
\item
 \emph{Treating the environment as a fault } does not work. In particular, if we treat the environment as a fault then the transitions from state $fs_4$ to $fs_5$ and from $fs_5$ to $fs_6$ in Figure \ref{preCookerFig} are not required to occur. If these actions do not occur, the overpressure valve is never be activated. Hence, neither the valve nor the vent mechanism reduces the pressure to be less than 4. 
Also, faults are expected to stop. By contrast, this is not the case with the environment actions. 
 
 \item \emph{Treating the environment transitions similar to program transitions} is also not acceptable. To illustrate this, consider the case where we want to make changes to the program in Figure \ref{preCookerFig}. For instance, if the overpressure valve is removed, then this would correspond to removing transition from $s_6$ (respectively $fs_6$) to where pressure is less than 4. Also, if we add another safety mechanism, it would correspond to adding new transitions. However, we cannot do the same with environment actions that capture the changes made by the heat source. For example, we cannot add new transitions (e.g., from $fs_4$ to $s_4$) to the environment, and we cannot remove transitions (e.g., from $s_4$ to $s_5$). In other words, even if we make any changes to the model in Figure \ref{preCookerFig} by adding or removing safety mechanisms, the transitions marked environment actions remain unchanged. We cannot introduce new environment transitions and we cannot remove existing environment transitions. This is what we mean by environment being \textit{unchangeable}.

\item \emph{ Treating the environment to be collaborative without some special fairness to the program} does not work either. In particular, without some special fairness for the program, the system can cycle through states $s_4, s_5, s_4, s_5\cdots$.

\item
 \emph{Treating the environment to be simultaneously collaborative as well as adversarial where the program has some special fairness} enables one to ensure that this program achieves its desired goals. In particular, we need the environment to be {\em collaborative}, i.e., if it reaches a state where only environment actions can execute then one of them does execute. (Note that this requirement cannot be expected of faults.) This is necessary to ensure that system can transition from state $fs_4$ to $fs_5$ and from $fs_5$ to $fs_6$ which is essential for recovery to a state where pressure is less than 4.

We also need the program to have special fairness to require that it executes {\em faster} than the environment so that it does not execute in a cycle through states $s_4, s_5, s_4, \cdots$. (We will precisely define the notion of faster in Section \ref{pre:programg}.)
\end{itemize}

{\bf Goal of the paper. } Based on the above example, our goal in this paper is to evaluate how such simultaneously collaborative and adversarial environment can be used in adding stabilization, failsafe fault-tolerance, and masking fault-tolerance to a given program.

Intuitively, in stabilizing fault-tolerance, starting from an arbitrary state, the program is guaranteed to recover to its legitimate states. In failsafe fault-tolerance, in the presence of faults, the program satisfies the safety specification. In masking fault-tolerance, in addition to satisfying the safety specification, the program recovers to its legitimate states from where future specification is satisfied. Also, the results from this work are applicable for nonmasking fault-tolerance from \cite{ag93}.

We also note that the results in \cite{bka12} do not model environment actions. Using the framework in \cite{bka12} for the above example would require one to treat the environment actions to be fault actions. And, as discussed above, this leads to an unacceptable result.

{\bf Contributions of the paper. } \
The main results of this work are as follows:

\begin{itemize} 
\item We present two algorithms for addition of stabilization to an existing program. Of these, the first algorithm is designed for the case where the program is provided with {\em minimal fairness} (where the program is given a chance to execute at least once between any two environment actions). The second algorithm, proposed in the Appendix, is for the case where additional fairness is provided. This algorithm is especially applicable when adding stabilization with minimal fairness is impossible. Both these algorithms are sound and complete, i.e., the program found by them is guaranteed to be stabilizing and if they declare failure then it implies that adding stabilization to that program is impossible. 
\item We present an algorithm for addition of failsafe fault-tolerance. This algorithm is also sound and complete.
\item We present an algorithm for addition of masking fault-tolerance. This algorithm is also sound and complete.
\item We note that the algorithm for masking fault-tolerance can be easily applied for designing nonmasking fault-tolerance discussed in \cite{ag93}. 
\item We show that the complexity of all algorithms presented in this paper is polynomial (in the state space of the program). Also, we note that the algorithms for stabilizing and masking fault-tolerance require one to solve the problem in a completely different fashion when compared to the case where we have no unchangeable environment actions. 
\end{itemize}

{\bf Organization of the paper. } \
This paper is organized as follows: in Section \ref{sec:pre} we provide the definitions of a program design, specifications, faults, fault-tolerance, and safe stabilization. In Section \ref{sec:stabilization} we define the problem of adding safe stabilization, and propose an algorithms to solve that problem for the case of minimal fairness. (The algorithm for the case where additional fairness is provided is proposed in the Appendix.) In Section \ref{sec:case}, as a case study, we illustrate how adding stabilization algorithm can be used for the controller of a smart grid. In Section \ref{sec:addft} we define the problem of adding fault-tolerance, and propose two algorithms to add failsafe and masking fault-tolerance. In Section \ref{sec:extensions} we show how our proposed algorithms can be extended to solve related problems. In Section \ref{sec:related}, we discuss related work. In section \ref{sec:appli}, we discuss application of our algorithms for cyber-physical and distributed systems. Finally, we make concluding remarks in Section \ref{sec:concl}

\section{Preliminaries}
\label{sec:pre}
In this section, we define the notion of programs, faults, specification and fault-tolerance. We define programs in terms of their states and transitions. The definitions of specification is based on that by Alpern and Schneider \cite{as85}. And, the definitions of faults and fault-tolerance are adapted from that by Arora and Gouda \cite{ag93}.

\subsection{Program Design Model}
\label{pre:programg}
\begin{definition}[Program]
A program $p$ is of the form $\langle S_p, \delta_p\rangle$ where $S_p$ is the state space of program $p$, and $\delta_p \subseteq S_p \times S_p$.
\end {definition}

The environment in which the program executes also changes the state of the program. Instead of modeling this in terms of concepts such as variables that are written by program and variables that are written by the environment, we use a more general approach where models it as a subset of $S_p \times S_p$. Thus, 

\begin{definition} [Environment]
An environment $\delta_e$ for program $p$, is defined as a subset of $S_p \times S_p$.

\end {definition}

\begin{definition} [State Predicate]
A state predicate of $p$ is any subset of $S_p$.
\end{definition}

\begin{definition}[Projection]
The projection of program $p$ on state predicate $S$, denoted as $p|S$, is the program $\langle S_p, \{(s_0,s_1): (s_0,s_1) \in \delta_p \wedge s_0, s_1 \in S\}\rangle$. In other words, $p|S$ consists of transitions of $p$ that start in $S$ and end in $S$. We denote the set of transitions of $p|S$ by $\delta_p|S$. 

\end{definition}

\begin{definition} [\ensuremath{p[]_k\delta_e\ computation}]
\label{def:progenvcomp}
Let $p$ be a program with state space $S_p$ and transitions $\delta_p$. Let $\delta_e$ be an environment for program $p$ and $k$ be an integer greater than 1. We say that a sequence $\langle s_0, s_1, s_2, ...\rangle$ is a $p[]_k\delta_e$ computation \textit{iff}
\begin{list}{\labelitemi}{\leftmargin=0.5em}
\itemsep=-2mm \vspace*{-2mm}
\item $\forall i : i \geq 0 : s_i \in S_p$, and
\item $\forall i : i \geq 0 : (s_i, s_{i+1}) \in \delta_p \cup \delta_e$, and
\item $\forall i : i \geq 0 : ((s_i, s_{i+1}) \in \delta_e) \Rightarrow$ \newline $(\forall l : i<l<i+k : (\exists s'_l:: (s_l, s'_l) \in \delta_p) \Rightarrow (s_l, s_{l+1}) \in \delta_p))$.
\end{list}

\end{definition}

Note that the above definition requires that in every step, either a program transition or an environment transition is executed. Moreover, after the environment transition executes, the program is given a chance to execute in the next $k\!-\!1$ steps. However, in any state that no program transition is available, an environment transition can execute. 

\begin{definition}
[Closure]
A state predicate $S$ is closed in a set of transitions $\delta$ iff $(\forall(s_0,s_1): (s_0,s_1)\in \delta: (s_0 \in S \Rightarrow s_1 \in S))$.
\end{definition}

\subsection{Specification}

Following Alpern and Schneider [7], we let the specification of program to consist of a safety specification and a liveness specification.

\begin{definition} [Safety]
The safety specification is specified in terms of a set of transitions, $\delta_b$, that the  program  is  not  allowed  to  execute.  Thus,  a  sequence
$\sigma = \langle s_0,s_1 , \ldots\rangle$ refines the safety specification $\delta_b$ iff  $\forall j : 0\!<\!j\!<\! length(\sigma) : (s_j , s_{j+1}) \notin \delta_b$.
\end{definition}

\begin{definition} [Liveness]
The liveness specification is specified in terms of a leads-to property ($L \leadsto T$) to denote, where both $L$ and $T$  are state predicates. Thus, a sequence $\sigma = \langle s_0,s_1 , \ldots\rangle$ refines the liveness specification iff $\forall j :$ $L$ is true in $s_j :  (\exists k: j \leq k < length(\sigma) :$ $T$ is true in $s_k)$.
\end{definition}

\begin{definition}
A specification, is a tuple $\langle Sf, Lv  \rangle$, where $Sf$ is a safety specialization and $Lv$ is a liveness specification. A  sequence  $\sigma$ satisﬁes  $spec$  iff it refines $Sf$ and $Lv$.
\end{definition}

\begin{definition}
[Refines]
$p[]_k \delta_e$ refines $spec$ from $S$ iff the following conditions hold:
\begin{list}{\labelitemi}{\leftmargin=0.5em}
\itemsep=-2mm
\vspace*{-2mm}
\item $S$ is closed in $\delta_p \cup \delta_e$, and
\item Every computation of $p[]_k \delta_e$ that starts from a state in S refines $spec$.
\end{list}
\end{definition}
We note that from the above definition, it follows that starting from a state in $S$, execution of either a program action or an environment action results in a state in $S$. Transitions that start from a state in $S$ and reach a state outside $S$ will be modeled as faults (cf. Definition \ref{def:fault}). 

\begin{definition} [Invariant]
If $p$ refines $spec$ from $S$ and $S \neq \phi$, we say that $S$ is an invariant of $p$ for $spec$.
\end{definition}

\subsection{Faults and Fault-Tolerance}

\begin{definition}
[Faults]
\label{def:fault}
A fault for $p (=\langle S_p, \delta_p \rangle)$ is a subset of $S_p \times S_p$.
\end{definition}

\begin{definition} [\ensuremath{p[]_k\delta_e[]f\ computation}]
Let $p$ be a program with state space $S_p$ and transitions $\delta_p$. Let $\delta_e$ be an environment for program $p$, $k$ be an integer greater than 1, and $f$ be the set of faults for program $p$. We say that a sequence $\langle s_0, s_1, s_2, ...\rangle$ is a $p[]_k\delta_e []f$ computation \textit{iff}
\begin{list}{\labelitemi}{\leftmargin=0.5em}
\itemsep=-2mm
\vspace*{-3mm}
\item $\forall i : i \geq 0 : s_i \in S_p$, and 
\item $\forall i : i \geq 0 : (s_i, s_{i+1}) \in \delta_p \cup \delta_e \cup f$ , and \item $\forall i : i \geq 0 : (s_i, s_{i+1}) \in \delta_e \Rightarrow$ \newline$\forall l: i<l<i+k : (\exists s'_l:: (s_l, s'_l) \in \delta_p \Rightarrow (s_l, s_{l+1}) \in (\delta_p \cup f))$, and \item $\exists n : n \geq 0 : (\forall j : j > n: (s_{j-1}, s_j) \in (\delta_p \cup \delta_e))$.
\end{list}
\end{definition}

The definition of fault-span captures the boundary up to which program could be perturbed by faults. Thus,

\begin{definition}
[Fault-span]
$T$ is an $f$-span of $p[]_k\delta_e$ from $S$ \textit{iff}
\begin{list}{\labelitemi}{\leftmargin=0.5em}
\itemsep=-2mm
\vspace*{-3mm}
\item $S\Rightarrow T$, and
\item for every computation $\langle s_0, s_1, s_2, \ldots \rangle$ of $p[]_k\delta_e[]f$, where $s_0 \in S$, $\forall i: s_i \in T$.
\end{list}
\end{definition}

A failsafe fault-tolerant program ensures that safety property is not violated even if faults occur. In other words, we have

\begin{definition}
[failsafe f-tolerant]
$p[]_k\delta_e$ is failsafe $f$-tolerant to $spec$ (=$\langle{Sf, Lv}\rangle$) from $S$ \textit{iff} the following two conditions hold:
\begin{list}{\labelitemi}{\leftmargin=0.5em}
\itemsep=-2mm
\vspace*{-3mm}
\item $p[]_k\delta_e$ refines $spec$ from $S$, and
\item every computation prefix of $p[]_k\delta_e[]f$ that starts from $S$ refines $Sf$.
\end{list}

In addition to satisfying the safety property, a masking fault-tolerant program recovers to its invariant.

\end{definition}

\begin{definition}
[masking f-tolerant]
$p$ is masking $f$-tolerant to $spec$ from $S$ \textit{iff} the following two conditions hold:
\begin{list}{\labelitemi}{\leftmargin=0.5em}
\itemsep=-2mm
\vspace*{-3mm}
\item $p[]_k\delta_e$ is failsafe $f$-tolerant to $spec$, and
\item there exists $T$ such that (1) $T$ is an $f$-span of $p[]_k\delta_e$ from $S$ and (2) for every computation  $\sigma (=\langle s_0, s_1, s_2, \ldots \rangle)$ of $p[]_k\delta_e[]f$ that starts from a state in $S$ if there exists $i \!>\! 0$ such that $s_i \in T-S$, then there exists $j \!>\! i$ such that $s_j \in S$.
\end{list}

\end{definition}

Condition (2) above simply means that in any computation which starts in $S$, when the program leaves $S$, it should return back to $S$. 

We also define the notion of stabilizing programs. We extend the definition from \cite{d74} and \cite{ssDolev2000} by requiring a stabilizing program to satisfy certain safety property during recovery. We consider this generalized notion because it allows us to capture program restrictions (such as inability to change environment variables) and because it is useful in our design of algorithm for adding masking fault-tolerance. The traditional definition of stabilization is obtained by setting $\delta_b$ in the following definition to be the empty set. 

\begin{definition} [Safe Stabilization]
$p[]_k\delta_e$ is $\delta_b$-safe stabilizing for invariant $S$ \textit{iff} following conditions hold:
\begin{list}{\labelitemi}{\leftmargin=0.5em}
\itemsep=-2mm
\vspace*{-3mm}
\item $S$ is closed in $\delta_p \cup \delta_e$, and
\item for any $p[]_k\delta_e$ computation $\langle s_0, s_1, s_2, ... \rangle$ there does not exist $l$ such that $(s_l,s_{l+1}) \in \delta_b$, and
\item for any $p[]_k\delta_e$ computation $\langle s_0, s_1, s_2, ... \rangle$ there exists $l$ such that $s_l \in S$.

\end{list}

\end{definition}

\begin{remark}
The notion of safe stabilization has been viewed from different angles in the literature. In \cite{CDDLR13}, authors consider the case where the program reaches an acceptable states quickly and converges to legitimate states after a longer time. By contrast, our notion simply requires that certain transitions (that violate safety specification) cannot be executed during recovery. 
\end{remark}

\section{Addition of Safe Stabilization}
\label{sec:stabilization}

In this section, we present our algorithm for adding safe stabilization to an existing program. In Section \ref{sefAddition}, we identify the problem statement. In Section \ref{sec:stabkeq2}, we present our algorithm for the case where the parameter $k$ (that identifies the fairness between program and environment acitons) is set to $2$. Due to reasons of space, the algorithm for arbitrary value of $k$ is presented in the Appendix. 

\subsection{Problem Definition}
\label{sefAddition}

The problem for adding safe stabilization begins with a program $p$, its invariant $S$, and a safety specification $\delta_b$ that identifies the set of bad transitions. The goal is to add stabilization so that starting from an arbitrary state, the program recovers to $S$. Moreover, we want to ensure that during recovery the program does not execute any transition in $\delta_b$. Also, we want to make sure that the execution of environment actions cannot prevent recovery to $S$. Thus, the problem statement is as follows:

\fbox{\rule{0mm}{0mm}
\begin{minipage}{3in}
Given program $p$ with state space $S_p$ and transitions $\delta_p$, state predicate $S$, set of bad transitions $\delta_b$, environment $\delta_e$, and $k > 1$, identify $p'$ with state space $S_p$ such that:
\begin{list}{\labelitemi}{\leftmargin=0.5em}
\item $p'|S = p|S$
\item $p'[]_k\delta_e$ is $\delta_b$-safe stabilizing for invariant $S$
\end{list}
\end{minipage}
}

\subsection{Addition of Safe Stabilization}
\label{sec:stabkeq2}

In this section, we present an algorithm for the problem of addition of stabilization defined in the Section~\ref{sefAddition}. The algorithm proposed here adds stabilization for $k\!=\!2$. When $k\!=\!2$, the environment transition can execute immediately after any program transition. By contrast, for larger $k$, the environment transitions may have to wait until the program has executed $k\!-\!1$ transitions.  
Observe that if $\delta_b \cap \delta_e$ is nonempty then adding stabilization is impossible. This is due to the fact that if the program starts in a state where such a transition can execute then it can immediately violate safety. Hence, this algorithm (but not the algorithms for adding failsafe and masking fault-tolerance) assumes that $\delta_b \cap \delta_e = \phi$.
%

The algorithm for adding stabilization is as shown in Algorithm \ref{ALG:K2}. In this algorithm, $\delta'_p$ is the set of transitions of the final stabilizing program. Inside the invariant, the transitions should be equal to the original program. Therefore, in the first line, we set $\delta'_p$ to $\delta_p|S$. State predicate $R$ is the set of states such that every computation starting from $R$ has a state in $S$. Initially (Line \ref{k2:initialR}) $R$ is initialized to $S$. In each iteration, state predicate $R_p$ is the set of states that can reach a state in $R$ using a safe program transition, i.e., a transition not in $\delta_b$. In Line~\ref{k2:addRp} we add such program transitions to $\delta'_p$.

In the loop on Lines~\ref{k2:loopBegin}-\ref{k2:loopEnd}, we add more states to $R$. We add $s_0$ to $R$ (Line \ref{k2:addtoR}), whenever every computation starting from $s_0$ has a state in $S$. A state $s_0$ can be added to $R$ only when there is no environment transition starting from $s_0$ and going to state outside $R \cup Rp$. In addition to this condition, there should be at least one transition from $s_0$ that reaches $R$. The loop on Lines \ref{k2:mainLoopBegin}-\ref{k2:mainLoopEnd} terminates if no state is added to $R$ in the last iteration. Upon termination of the loop, the algorithm declares failure to add stabilization if there exists a state outside $R$. Otherwise, it returns $\delta'_p$ as the set of transitions of the stabilizing program. 

We use Figure \ref{addStab} to illustrate Algorithm~\ref{ALG:K2}. 
Figure \ref{addStab} depicts the status of the state space in a hypothetical $i^{th}$ iteration of loop on Lines \ref{k2:mainLoopBegin}-\ref{k2:mainLoopEnd}. In this iteration state \textbf{A} is added to $R$. This is due to the fact that (1) there is at least one transition from \textbf{A} (namely $(\textbf{A}, \textbf{F})$) that reaches $R$ and (2) there is no environment transition from \textbf{A} that reaches outside $R\cup R_p$. Likewise, state \textbf{C} is also added to $R$. 
State \textbf{B} is not added to $R$ due to environment transition $(\textbf{B}, \textbf{E})$. 
Likewise, state \textbf{D} is also not added to $R$. 
State \textbf{E} is not added to $R$ since there is no transition from \textbf{E} to a state in $R$. 

In the next, i.e., $(i+1)^{th}$, iteration, \textbf{E} is added to $R$ since there is a transition $(E, A)$ and $A$ was added to $R$ in the $i^{th}$ iteration. Continuing this, $D$ is added in the $(i+2)^{th}$ iteration.

\begin{figure}
\begin{center}
\includegraphics[width=80mm,scale=0.5]{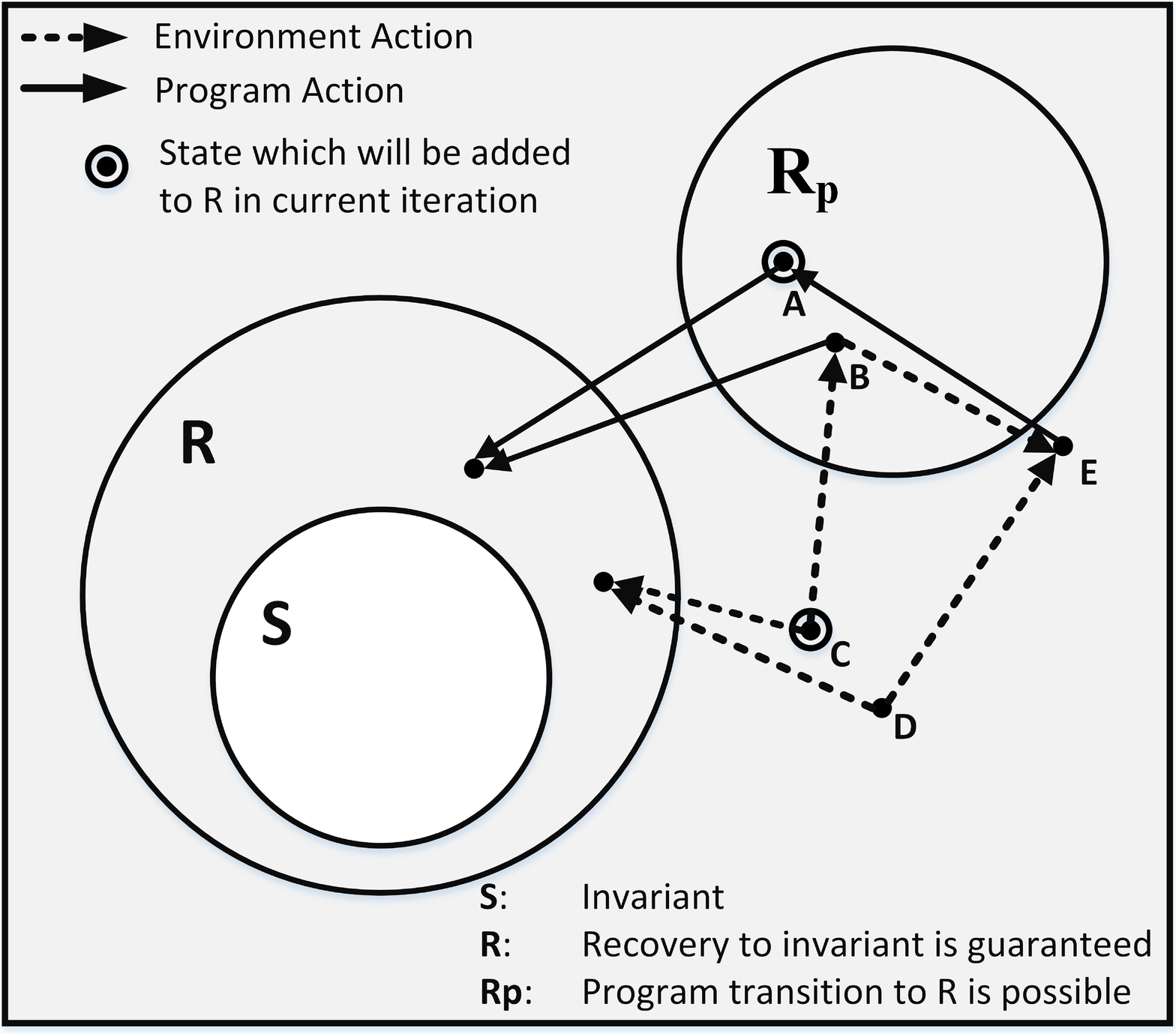}
\caption{Illustration of how R expands in Algorithm \ref{ALG:K2}}
\label{addStab}
\end{center}
\end{figure}

\begin{algorithm}
{
\caption{Addition of safe stabilization}
\label{ALG:K2}
\begin{algorithmic} [1]
\INPUT $S_p, \delta_p, \delta_e, S,$ and $ \delta_b$
\OUTPUT $\delta'_p$ or Not-Possilbe
\STATE $\delta'_p := (\delta_p | S);$
\STATE $R = S;$ \label{k2:initialR}
\REPEAT  \label{k2:mainLoopBegin}
\STATE  $R' = R;$
\STATE $Rp = \{s_0 | s_0 \notin R  \wedge   \exists s_1  : s_1 \in R  : (s_0, s_1) \notin \delta_b\};$
\FOR{$each \ s_0 \in  R_p$ }
\STATE \label{k2:addRp} $\delta'_p = \delta'_p \cup  \{(s_0, s_1) | (s_0, s_1) \notin \delta_b \wedge s_1 \in R \};$
\ENDFOR
\FOR{$each \ s_0 \notin R : \nexists s_2 \in \neg (R \cup R_p) : (s_0, s_2) \in \delta_e \wedge \newline (\exists s_1 : s_1 \in (R \cup R_p) : (s_0, s_1) \in \delta_e  \lor   s_0 \in R_p)$}  \label{k2:loopBegin}
\STATE \label{k2:addtoR} $R = R \cup {s_0};$	
\ENDFOR \label{k2:loopEnd}
\UNTIL \label{k2:mainLoopEnd} $(R' = R);$
\IF{$\exists s_0 \notin R $}
\RETURN  'Not-Possible';
\ELSE
\RETURN  $\delta'_p;$
\ENDIF

\end{algorithmic}
}
\end{algorithm}

\begin{theorem}
\label{thm:alg1}
Algorithm \ref{ALG:K2} is sound and complete. And, its complexity is polynomial.
\end{theorem}

For reasons of space, we provide the proofs in Appendix.

\section{Case Study: Stabilization of  Smart Grid}
\label{sec:case}
In this section we illustrate how Algorithm~\ref{ALG:K2} is used to add safe stabilization to a controller program of a smart grid. We consider an abstract version of the smart grid described in \cite{smartGrid} (see Figure \ref{smartGrid}). In this example, the system consists of a generator $G$ and two loads $Z_1$ and $Z_2$. There are three sensors in the system. Sensor G shows the power generated by the generator, and sensors 1 and 2 show the demand of load $Z_1$ and $Z_2$, respectively. The goal is to ensure that proper load shading is used if the load is too high (respectively, generating capacity is too low). 

\begin{figure} [H]
\begin{center}
\includegraphics[width=80mm,scale=0.5]{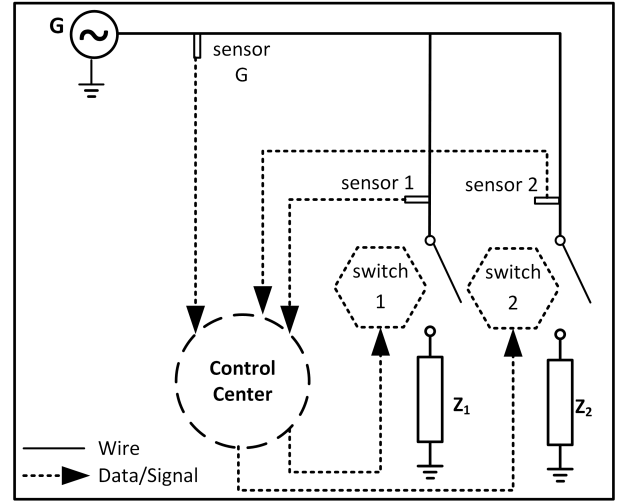}
\caption{Elementary single generator smart grid system}
\label{smartGrid}
\end{center}
\end{figure}

The control center is shown by a dashed circle in Figure \ref{smartGrid}. It can read the values of the sensors and turn on/off switches connected to the loads. The program of the control center should control switches in a manner that all the conditions below are satisfied: 
\begin{enumerate}
 \item Both switches should be turned on if the overall sensed load is less than or equal to the generation capacity. 
 \item If sensor values reveal that neither load can individually be served by G then both are shed.
 \item If only one load can be served then the smaller load is shed assuming the larger load can be served by G.
 \item  If only one can be  served and the larger load exceeds the generation capacity, the smaller load is served.
\end{enumerate}

\subsection{Program Model}
\label{sec:case:model}
We model the program of the smart grid shown in Figure.~\ref{smartGrid} by program $p$ which has five variables as follows: \\ \\
$V_G: $ The value of sensor G. \\
$V_1: $ The value of sensor 1. \\
$V_2: $ The value of sensor 2. \\
$w_1: $ The status of switch 1. \\
$w_2: $ The status of switch 2.

The value of each sensor is an integer in the range $[0, max]$. And, the status of each switch is a Boolean. 


The invariant $S$ for this program includes all the states which are legitimate according to the conditions 1-4 mentioned above. Therefore, $S$ is the union of state predicates $I_1$ to $I_6$ as follows \footnote{We need to add $0 \leq V_1, V_2, V_g \leq max$ to all conditions. For brevity, we keep these implicit.}:\\ \\
{\small
$I_1 = (V_1 + V_1 \leq V_G) \wedge (w_1  \wedge w_2 )$ \\
$I_2 = V_1 \leq V_G \wedge V_2 > V_G) \wedge  (w_1  \wedge \neg w_2 )$ \\
$I_3 =(V_1 > V_G \wedge V_2 \leq V_G) \wedge  (\neg w_1 \wedge w_2 )$ \\
$I_4 = (V_1 > V_G \wedge V_2 > V_G) \wedge  ( \neg w_1  \wedge \neg w_2 )$ \\
$I_5 = (V_1 + V_2 > V_G \wedge V_1 \leq V_G \wedge V_2 \leq V_G \wedge V_1 \leq V_2) \wedge \newline (\neg w_1 \wedge w_2 )$ \\
$I_6 = (V_1 + V_2 > V_G \wedge V_1 \leq V_G \wedge V_2 \leq V_G \wedge V_1 > V_2) \wedge \newline (w_1  \wedge \neg w_2 )$
}

\begin{observation}
\label{obs:forEveryE}
For \textit{any} value of $V_1$, $V_2$, and $V_G$, there exists an assignmet to $w_1$ and $w_2$ such that the resulting state is in $S$. 

\end{observation}

The values of sensors can change by environment transitions. In addition, environment can keep the current value of a sensor by self-loop environment transitions. However, environment cannot change the status of switches. Thus, set of environment transitions, $\delta_e$ is equal to $\{(s_0, s_1)| \big(w_1(s_0) = w_1(s_1)\big) \wedge \big(w_2(s_0)= w_2(s_1)\big) \}$, where $w_i(s_j)$ shows the status of the switch $i$ in state $s_j$.

Program cannot change the value of any sensor. Thus, set of bad transitions, $\delta_b$ for this program is equal to $\{(s_0, s_1)| \ V_G(s_0) \neq V_G(s_1) \vee V_1(s_0) \neq V_1(s_1) \vee V_2(s_0) \neq V_2(s_1)\}$, where $V_i(s_i)$ shows the value of the variable $V_i$ in state $s_i$.

For the sake of presentation and to illustrate the role of $k$, we also assume that program cannot change the status of more than one switch in one transition. For this case, we add more transitions to the set of bad transitions. We call the set of bad transitions for this case $\delta_{b_2}$ and it is equal to  $\{(s_0, s_1)| \ V_G(s_0) \neq V_G(s_1) \vee V_1(s_0) \neq V_1(s_1) \vee V_2(s_0) \neq V_2(s_1) \vee (w_1(s_0) \neq w_1(s_1) \wedge w_2(s_0) \neq w_2(s_1))\}$.

\subsection{Adding Stabilization}
Here, we apply Algorithm~\ref{ALG:K2} to add stabilization to program $p$ defined in Section \ref{sec:case:model}. We illustrate the result of applying Algorithm~\ref{ALG:K2} for two sets of bad transitions, $\delta_b$ and $\delta_{b_2}$. 

\subsubsection{Adding Stabilization for $\delta_b$ }

At the beginning of Algorithm~\ref{ALG:K2}, $R$ is initialized with $S$. In the first iteration of loop on Lines~\ref{k2:mainLoopBegin}-\ref{k2:mainLoopEnd}, $R_p$ is the set of states outside $S$ that can reach a state in $S$ with only one program transition. A program transition cannot change the value of any sensor. 

According to Observation~\ref{obs:forEveryE}, from each state in $\neg S$ it is possible to reach a state in $S$ with changing the status of switches. Therefore, following set of transitions are added to $\delta'_p$ by Line~\ref{k2:addRp}: \\ \\
$\{(s_0,s_1)| \  V_1(s_0) = V_1(s_1) \wedge V_2(s_0) = V_2(s_1) \wedge V_G(s_0) = V_G(s_1)\wedge s_0 \notin \cup_{i=1}^6I_i \wedge  s_1 \in \cup_{i=1}^6I_i\}$

Since every state in $\neg S$ ($\neg R$) is in $R_p$, there does not exist any environment transition starting from any state to a state in $\neg (R \cup R_p)$. Therefore, all the states in $\neg R$ are added to $R$ by Line~\ref{k2:addtoR}. 

In the second iteration no more states are added to $R$. Thus, loop on Line~\ref{k2:mainLoopBegin}-\ref{k2:mainLoopEnd} terminates. Since there is no state in $\neg R$, the algorithm returns $\delta'_p$ as the transition of the resulting $\delta_b$-safe stabilizing program for $S$.
\subsubsection{Adding Stabilization for $\delta_{b_2}$ }

At the beginning of Algorithm~\ref{ALG:K2}, $R$ is initialized with $S$. In the first iteration of loop on Lines~\ref{k2:mainLoopBegin}-\ref{k2:mainLoopEnd}, $R_p$ is the set of states outside $S$ that can reach a state in $S$ with only one program transition. A program transition cannot change the value of any sensor. In addition, according to $\delta_{b_2}$, it cannot change the status of both switches. Therefore, state predicate $R_p$ is the union of state predicates $R_{p_1}$ to $R_{p_6}$ as follows ($\oplus$ denotes the xor operation):\\ \\ 
{\small
$R_{p_1} = (V_1 + V_1 \leq V_G) \wedge (w_1 \oplus w_2) $ \\
$R_{p_2} = (V_1 \leq V_G \wedge V_2 > V_G)  \wedge (w_1 \oplus \neg w_2) $ \\
$R_{p_3} = (V_1 > V_G \wedge V_2 \leq V_G) \wedge (\neg w_1 \oplus  w_2) $ \\
$R_{p_4} = (V_1 > V_G \wedge V_2 > V_G) \wedge (\neg w_1 \oplus \neg w_2) $ \\
$R_{p_5} = (V_1 + V_2 > V_G \wedge V_1 \leq V_G \wedge V_2 \leq V_G \wedge V_1 \leq V_2) \wedge (\neg w_1 \oplus  w_2) $ \\
$R_{p_6} = (V_1 + V_2 > V_G \wedge V_1 \leq V_G \wedge V_2 \leq V_G \wedge V_1 > V_2) \wedge (w_1 \oplus \neg w_2) $\\
}

Similarly, $\neg (R \cup R_p)$ includes every state that is outside $S$ and more than one step is needed to reach a state in $S$. Therefore, state predicate $\neg (R \cup R_p)$ is the union of state predicates $R'_{p_1}$ to $R'_{p_6}$ as follows:\\ \\
{\small
$R'_{p_1} = (V_1 + V_1 \leq V_G) \wedge (\neg w_1 \wedge \neg w_2) $ \\
$R'_{p_2} = (V_1 \leq V_G \wedge V_2 > V_G) \wedge (\neg w_1 \wedge  w_2) $ \\
$R'_{p_3} = (V_1 > V_G \wedge V_2 \leq V_G) \wedge ( w_1 \wedge  \neg w_2)$ \\
$R'_{p_4} = (V_1 > V_G \wedge V_2 > V_G) \wedge ( w_1 \wedge  w_2) $ \\
$R'_{p_5} = (V_1 + V_2 > V_G \wedge V_1 \leq V_G \wedge V_2 \leq V_G \wedge V_1 \leq V_2) \wedge ( w_1 \wedge  \neg w_2) $ \\
$R'_{p_6} = (V_1 + V_2 > V_G \wedge V_1 \leq V_G \wedge V_2 \leq V_G \wedge V_1 > V_2) \wedge (\neg w_1 \wedge  w_2) $\\
}

Now, observe that for any status of switches, there exists a state in $\neg (R \cup R_p)$. That means from any state in $S_p$ it is possible to reach a state in  $\neg (R \cup R_p)$ without changing the value of switches using an environment transition. Therefore, no state is added to $R$ in the first iteration, and loop on Lines~\ref{k2:mainLoopBegin}-\ref{k2:mainLoopEnd} terminates in the first iteration. Since, all the states outside $S$ remains in $\neg R$, the algorithm declares no solution to the addition problem exists. Therefore, according to the completeness of the Algorithm~\ref{ALG:K2}, there does not exist any $\delta_{b_2}$-safe stabilizing program for the smart grid described in this section when $k$ is equal to 2.  This is expected since the only solution for this problem requires changing both sensors simultaneously before the environment is able to disrupt it again. This program does have a solution for $k\!=\!3$. But we omit its derivation for lack of space. 
\section{Addition of Fault-Tolerance}
\label{sec:addft}
In this section, we present our algorithm for adding failsafe and masking fault-tolerance. 
In Section \ref{sec:probdef}, we identify the problem statement for adding these levels of fault-tolerance. In Section \ref{sec:failsafe}, we present our algorithm for adding failsafe fault-tolerance. Section \ref{sec:masking} presents an algorithm for adding masking fault-tolerance. Finally, we show that the same algorithm can be used for adding nonmasking fault-tolerance considered in \cite{ag93}.

\subsection{Problem Definition}
\label{sec:probdef}

In addition to the set of bad transitions $\delta_b$ that we used for providing safe stabilization, in this case, we introduce additional parameter $\delta_r$ that identifies additional restrictions on program transitions. As an example, consider the case where a program cannot change the value of sensor, i.e., it can only read it. However, the environment can change the value of the sensor. In this case, transitions that change the value of the sensor are disallowed as program transitions but, they are acceptable as environment transitions. Note that this was not necessary in Section \ref{sec:stabilization} since we could simply add these transitions to $\delta_b$, i.e., transitions that violate safety. This is acceptable since addition stabilization requires $\delta_b \cap \delta_e \!=\! \phi$. However, adding failsafe or masking fault-tolerance is possible even if $\delta_b \cap \delta_e \neq \phi$. Hence, we add the parameter $\delta_r$ explicitly. The problem statement for addition of fault-tolerance is as follows:

\fbox{\rule{0mm}{0mm}
\begin{minipage}{3in}
Given $p$, $\delta_e$, $S$, $spec$, set of program restrictions $\delta_r$, $k > 1$, and $f$ such that $p[]_k\delta_e$ refines $spec$ from $S$, and $\delta_p \cap \delta_r = \phi$, identify $p'$ and $S'$ such that:

\begin{list}{\labelitemi}{\leftmargin=0.5em}
\item \textbf{\textit{C1}}: every computation of $ p'[]_k\delta_e$ that starts in a state in $S'$ is a computation of $ p[]_k\delta_e$ that starts in $S$, and

\item \textbf{\textit{C2}}: $p'[]_k\delta_e$ is failsafe (respectively, masking) $f$-tolerant to $spec$ from $S'$ and
\item \textbf{\textit{C3}}: $\delta'_p \cap \delta_r = \phi$

\end{list}
\end{minipage}
}

The problem statement requires that the program does not introduce new behaviors in the absence of faults (Constraint $C1$), provides desired fault-tolerance (Constraint $C2$), and does not include a transition in $\delta_r$ (Constraint $C3$).

\begin{assumption}
\label{assumption:infinite}
For simplicity of the algorithms and its proof, we assume that there are no deadlocks in $\delta_p [] \delta_e$ in any state in $S$. In other words, for any $s_0$ in $S$, there exists a state $s_1$ in $S$ such that $(s_0, s_1)$ is in $\delta_p \cup \delta_e$. If this is not true then we can add self-loops corresponding to those states, i.e., states in $\{ s_0 | s_0 \in S \wedge \forall s_1 :: (s_0, s_1) \not \in \delta_p \cup \delta_e \}$. Finally, after the fault-tolerant program is obtained, we remove these self-loops. We note that this does not affect either soundness or completeness of any of our algorithms. 
\end{assumption}

\subsection{Adding Failsafe Fault-Tolerance}
\label{sec:failsafe}

The algorithm for adding failsafe fault-tolerance for $k\!=\!2$ is as shown in Algorithm \ref{ALG:FAILSAFE}. In this algorithm set $ms_1$ is the set of states no matter how they are reached, starting from them, there exists a computation suffix which violates safety. Set $ms_2$ is the set of states if they are reached by a program or fault transition, starting from them, there exists a computation suffix which violates safety. Note that $ms_2$ always includes $ms_1$. Initially, $ms_1$ is initialized to $\{s_0| (s_0, s_1) \in  f \cap \delta_b\}$, and $ms_2$ is initialized to $ms_1 \cup \{s_0 | \exists s_1:: (s_0, s_1) \in \delta_e \cap \delta_b\}$ by Lines \ref{fai:ini1} and \ref{fai:ini2}. Set $mt$ is the set of transitions that the final program cannot have, as they are in $\delta_b \cup \delta_r$, or reach a state in $ms_2$. 

In the loop on Lines \ref{fail:firstLoopBeging} - \ref{fail:firstLoopEnd}, more states are added to $ms_1$ and $ms_2$. Consequently, $mt$ should be updated. Any state $s_0$ is added to $ms_1$ by Line \ref{fai:addtoMs1} in two cases: 1) if there exists a fault transition starting from $s_0$ that reaches a state in $ms_2$ 2) if there exists an environment transition $(s_0, s_1)$ such that $(s_0, s_1)$ is a bad transition or $s_1 \in ms_1$, and any transition starting from $ms_1$ reaches a state in $ms_2$ (i.e., any transition $(s_0, s_2) \in mt$). 

A state is added to $ms_2$ by Line \ref{fai:addtoMs2} if it is added to $ms_1$ or if there exists an environment transition to a state in $ms_1$. We update $mt$ by Line \ref{fai:updateMt} to include transitions to new states added to $ms_2$.The loop on Lines \ref{fail:firstLoopBeging} - \ref{fail:firstLoopEnd} terminates if no state is added to $ms_1$ or $ms_2$ in an iteration. 

Then, we focus on creating new invariant, $S'$, for the revised program. $S'$ cannot include any transition in $ms_2$, as starting from any state in $ms_2$, there is a computation which violates safety. In addition, the set of program transitions of the revised program, $\delta'_p$, cannot include any transition in $mt$, as by any transition in $mt$ a state in $ms_2$ is reached. Thus, we initialized $\delta'_p$ with $\delta_p|S - mt$. Note that $S'$ should be closed in $p'[]_2\delta_e$. In addition, according to Assumption~\ref{assumption:infinite}, $S'$ cannot include any deadlock state. Thus, anytime that we remove a state from $S'$ we ensure these condition by calling $RemoveDeadlock$ and $EnsureClosure$ functions. 

Note that according to condition \textit{C1}, of the addition problem defined in the Section \ref{sec:probdef}, the set of computations of the revised program inside its invariant should be a subset of set of computations of the original program inside its invariant. Thus, the revised program cannot have any new computation starting from its invariant. In loop on Lines \ref{fail:removems4Begin} -  \ref{fail:removems4End} we remove states from $S'$ to avoid creating such new computations. 

\begin{algorithm}
{
\caption{Adding Failsafe Fault-Tolerance}
\label{ALG:FAILSAFE}
\begin{algorithmic} [1]
\INPUT $S_p, \delta_p, \delta_e, S, \delta_b$, $\delta_r$, $k$, and $f$
\OUTPUT $(\delta'_p, S')$ or Not-possilbe

\STATE \label{fai:ini1} $ms_1 = \{s_0| (s_0, s_1) \in  f \cap \delta_b\};$
\STATE \label{fai:ini2} $ms_2 = ms_1 \cup \{s_0 | \exists s_1:: (s_0, s_1) \in \delta_e \cap \delta_b\};$
\STATE $mt =   \{(s_0, s_1) | \ (s_0 ,s_1) \in  (\delta_b \cup \delta_r) \vee  s_1 \in ms_2\};$
\REPEAT \label{fail:firstLoopBeging}
\STATE $ms_1' = ms_1;$
\STATE $ms_2' = ms_2;$

\STATE \label{fai:addtoMs1} $ms_1 = ms_1 \cup \{s_0 | \ \exists s_1: s_1 \in  ms_2: (s_0 ,s_1) \in f \} \cup \{s_0 |(\exists s_1:: (s_1 \in ms_1 \wedge (s_0, s_1) \in \delta_e) \vee (s_0, s_1) \in \delta_e \cap \delta_b )) \wedge (\forall s_2:: (s_0, s_2) \in mt)\};$
\STATE \label{fai:addtoMs2}$ms_2 = ms_2 \cup ms_1 \cup \{s_0 | \exists s_1: s_1 \in ms_1 : (s_0, s_1) \in \delta_e) \};$
\STATE \label{fai:updateMt} $mt =   \{(s_0, s_1) | \ (s_0 ,s_1) \in  (\delta_b \cup \delta_r) \vee  s_1 \in ms_2\};$

\UNTIL \label{fail:firstLoopEnd} {$(ms_1' = ms_1 \wedge ms_2' = ms_2  )$}
\STATE $\delta'_p = \delta_p |S - mt;$
\STATE $S' = RemoveDeadlock(S - ms_2,  \delta'_p, \delta_e);$
\REPEAT \label{fail:removems4Begin}
\IF {$S' =\phi$}
\RETURN \label{fail:isEmpty} Not-possible;
\ENDIF
\STATE $S'' = S';$
\STATE $\delta'_p = EnsureClosure(\delta'_p, S');$
\STATE $ms_3 = \{s_0 | \  \big ( \exists s_1, s_2:: (s_0, s_1) \in \delta_e \ \wedge (s_0, s_2) \in \delta_p \ \big ) \wedge$ $
                      \big ( \nexists s_3:: (s_0, s_3) \in \delta'_p  \big ) \};$
\STATE $ms_4 = \{ s_0| \exists s_1 :: (s_1 \in ms_3 \wedge (s_0, s_1) \in \delta_e) \}$
\STATE $S' = RemoveDeadlock (S' - ms_4, \delta_p , \delta_e)$
\UNTIL ($S'' = S'$)  \label{fail:removems4End}
\STATE \label{fail:line:finaldelta} $\delta'_p = \Big(\delta'_p \cup \big ( ( S_p - S') \times S_p )\big ) \Big ) -mt ;$
\RETURN $(\delta'_p, S')$ ;
\newline
\STATE $RemoveDeadlock (S, \delta_p , \delta_e)$
\STATE \hspace*{3.5mm} \textbf{repeat}
\STATE \hspace*{7mm} $S' = S;$
\STATE \hspace*{8.5mm}$S = S - \{s_0| \ (\forall s_1: s_1 \in S: (s_0,s_1) \notin  \delta_p) \}; $
\STATE \hspace*{8.5mm}$S = S - \{s_0 | \ \exists s_1 :: (s_0,s_1) \in \delta_e \wedge s_0 \in S \wedge s_1 \notin S\};$
\STATE \hspace*{3.5mm} \textbf{until} $(S' = S)$
\STATE \hspace*{3.5mm} \textbf{return} $S;$
\newline
\STATE $EnsureClosure (p, S)$
\STATE \hspace*{3.5mm} \textbf{return} $p-\{(s_0,s_1):: s_0 \in S \wedge s_1 \notin S\};$

\end{algorithmic}
}
\end{algorithm}

Consider a state $s_0$ starting from which there exists environment transition $(s_0, s_1)$. In addition there exists program transition $(s_0, s_2)$ in the set of program transitions of the original program, $\delta_p$. Set $ms_3$ includes any state like $s_0$. If $s_0$ is reached by environment transition $(s_3,s_0)$, in the original program according to fairness assumption, $(s_0, s_1)$ cannot occur. Thus, sequence $\langle s_3, s_0, s_1\rangle$ cannot be in any computation of $p[]_2\delta_e$. However, if we remove program transition $(s_0, s_2)$ in the revised program, $\langle s_3, s_0, s_1\rangle$ can be in computation of $p'[]_2\delta_e$. Therefore, we should remove any state like $s_3$ from the invariant. Set $ms_4$ includes any state like $s_3$.

After creating invariant $S'$, we add program transitions outside it to $\delta'_p$. Note that outside $S'$, any program transition which is not in $mt$ is allowed to exists in the final program. In Line~\ref{fail:isEmpty}, the algorithm declare the no solution to the addition problem exists, if $S'$ is empty. Otherwise, at the end of the algorithm, it returns $(\delta'_p, S')$ as the solution to the addition problem.

\begin{theorem}
\label{thm:fsMainTheorem}
Algorithm \ref{ALG:FAILSAFE} is sound and complete. And, its complexity is polynomial.
\end{theorem}

For reasons of space, we provide the proofs in Appendix.

\subsection{Adding Masking Fault-Tolerance}
\label{sec:masking}

In this section, we present the algorithm for adding masking fault-tolerance in the presence of unchangeable environment actions. The intuition behind this algorithm is as follows: First, we utilize the ideas from adding stabilization. Intuitively, in Algorithm \ref{ALG:K2}, we constructed the set $R$ from where recovery to invariant ($S$) was possible. In case of stabilization, we wanted to ensure that $R$ includes all states. However, for masking, this is not necessary. Also, the algorithm for adding stabilization does not use faults as input. Hence, we need to ensure that recovery from $R$ is not prevented by faults. This may require us to prevent the program from reaching some states in $R$. Hence, this process needs to be repeated to identify a set $R$ such that both recovery to $S$ is provided and faults do not cause the program to reach a state outside $R$. In addition, in masking fault-tolerance, like failsafe fault-tolerance, program should refine safety of spec even in presence of faults. Thus, the details of the Algorithm~\ref{ALG:MASKING} are as follows:

In this algorithm, both Algorithm~\ref{ALG:K2} and Algorithm~\ref{ALG:FAILSAFE} with some modification are used in the loop on Lines~\ref{mas:outerLoopBegin}-\ref{mas:outerLoopEnd}. 
First, in the loop on Lines \ref{mas:stBegin}-\ref{mas:stEnd} we build set $R$ which include all states from which all computations reach a state in $S$. In addition, all required program transitions are added to $\delta'_p$ by Line~\ref{mas:addToDeltaP}. 
When loop on Lines \ref{mas:stBegin}-\ref{mas:stEnd} terminates, we set $ms_1$ to $\neg (R \cup R_p)$, because a state in $\neg (R \cup R_p)$ should not be reached by any program, fault, or environment transition. We also set $ms_2$ to $\neg R$, as a state in $\neg R$ should not be reach by any fault or program transition. Then by Lines \ref{mas:expandingMsBeign}-\ref{mas:expandingMsEnd} we expand $ms_1$ and $ms_2$ with the same algorithm in the Algorithm~\ref{ALG:FAILSAFE}. In Line~\ref{mas:removeMT}, we remove any transition in $mt$ from $\delta'_p$, as any transition in $mt$ reaches a state in $ms_2$, and a program transition should not reach a state in $ms_2$.

In the loop on Lines~\ref{mas:ms4Begin}-\ref{mas:ms4End}, we remove some states from $S'$ to avoid new behavior inside the invariant just like we did in Algorithm~\ref{ALG:FAILSAFE}. If any state $s_0$ is removed from $S'$ in Lines \ref{mas:removeMs2} or \ref{mas:removeMs4}, we need to repeat the loop on Lines~\ref{mas:outerLoopBegin}-\ref{mas:outerLoopEnd}, because it is possible that a state in $R$ was dependent on $s_0$ to reach $S'$, but $s_0$ is not in $S'$ anymore. In Line~\ref{mas:isEmpty}, the algorithm declares that there does not exist a solution if $S'$ is empty. Otherwise, when loop on Lines~\ref{mas:outerLoopBegin}-\ref{mas:outerLoopEnd} terminates, the algorithm returns $(\delta'_p, S')$ as the solution to the addition problem.

\begin{theorem}
\label{thm:mMainTheorem}
Algorithm \ref{ALG:MASKING} is sound and complete. And, its complexity is polynomial.
\end{theorem}

For reasons of space, we provide the proofs in Appendix. 

Finally, we note that the addition of nonmasking fault-tolerance considered \cite{ag93} is also possible with Algorithm \ref{ALG:MASKING}. In particular, in this case, we need to set $\delta_b$ to be the empty set. In principle, Algorithm \ref{ALG:MASKING} could also be used to add stabilization. However, we presented Algorithm \ref{ALG:K2} separately since it is much simpler algorithm and forms the basis of Algorithm \ref{ALG:MASKING}. Moreover, Algorithm \ref{ALG:K2} can be extended to arbitrary value of $k$ (as done in Algorithm \ref{alg:GeneralK} in the Appendix). However, the corresponding problem of failsafe and masking fault-tolerance is open. In particular, Algorithms \ref{ALG:FAILSAFE} and \ref{ALG:MASKING} are sound even if we use an arbitrary value of $k$. However, they are not complete.

\section{Extensions of Algorithms}
\label{sec:extensions}

In this section, we consider problems related to those addressed in Sections \ref{sec:stabilization} and \ref{sec:addft}. Our first variation focuses on Definition \ref{def:progenvcomp}. In this definition, we assumed that the environment is {\em fair}. Specifically, at least $k\!-\!1$ actions execute between any two environment actions. We consider variations where (1)  this property is satisfied eventually. In other words, for some initial computation, environment actions may prevent the program from executing. However, eventually, fairness is provided to program actions, and (2) program actions are given even reduced fairness. Specifically, we consider the case where several environment actions can execute in a row but program actions execute infinitely often.

\begin{algorithm}
{
\caption{Adding Masking Fault-Tolerance}
\label{ALG:MASKING}
\begin{algorithmic} [1]
\STATE $S' = S;$
\STATE $\delta'_p = (\delta_p | S);$
\REPEAT \label{mas:outerLoopBegin}
\STATE $R = S';$
\STATE $S'' = S';$
\REPEAT \label{mas:stBegin}
\STATE $R' = R;$
\STATE $Rp = \{s_0 | s_0 \notin R  \wedge   \exists s_1 : s_1\in R  : (s_0, s_1) \notin (\delta_b \cup \delta_r) \};$
\FOR{$each \ s_0 \in  R_p$ }
\STATE  \label{mas:addToDeltaP} $\delta'_p = \delta'_p \cup  \{(s_0, s_1) | (s_0, s_1) \notin (\delta_b \cup \delta_r) \wedge s_1 \in R \};$
\ENDFOR
\FOR{$each \ s_0 \notin R : \nexists s_2 : s_2 \in  \neg(R \cup R_p) : (s_0, s_2) \in \delta_e \wedge  (\exists s_1 \in R \cup R_p : (s_0, s_1) \in \delta_e  \lor   s_0 \in R_p)$}
\STATE $R = R \cup {s_0};$	
\ENDFOR
\UNTIL \label{mas:stEnd} $(R' = R);$

\STATE $ms_1 = \neg (R \cup R_p);$
\STATE $ms_2 = \neg R;$

\STATE \label{mas:expandingMsBeign} $ms_1 = ms_1 \cup \{s_0| (s_0, s_1) \in  f \cap \delta_b\};$
\STATE $ms_2 = ms_2 \cup ms_1 \cup \{s_0 | \exists s_1:: (s_0, s_1) \in \delta_e \cap \delta_b\};$
\STATE $mt =   \{(s_0, s_1) | \ (s_0 ,s_1) \in  (\delta_b \cup \delta_r) \vee  s_1 \in ms_2\};$
\REPEAT
\STATE $ms_1' = ms_1;$
\STATE $ms_2' = ms_2;$

\STATE \label{masking:addtoMS1} $ms_1 = ms_1 \cup \{s_0 | \ \exists s_1: s_1 \in  ms_2: (s_0 ,s_1) \in f \} \cup \{s_0 |\big(\exists s_1: s_1 \in ms_1 : (s_0, s_1) \in \delta_e) \vee (s_0, s_1) \in (\delta_e \cap \delta_b )\big) \wedge \big(\nexists s_2:: (s_0, s_2) \in (\delta'_p - mt)\big)\};$
\STATE $ms_2 = ms_2 \cup ms_1 \cup \{s_0 | \exists s_1: s_1 \in ms_1 : (s_0, s_1) \in \delta_e) \};$
\STATE $mt =   \{(s_0, s_1) | \ (s_0 ,s_1) \in  (\delta_b \cup \delta_r) \vee  s_1 \in ms_2\};$

\UNTIL \label{mas:expandingMsEnd}{$ms_1' = ms_1 \wedge ms_2' = ms_2  $}
\STATE \label{mas:removeMT}$\delta'_p = \delta'_p - mt;$
\STATE \label{mas:removeMs2} $S' = RemoveDeadlock(S - ms_2,  \delta'_p, \delta_e);$
\REPEAT \label{mas:ms4Begin}
\IF {$S' =\phi$}
\RETURN \label{mas:isEmpty} Not-possible;
\ENDIF
\STATE $S''' = S';$
\STATE $\delta'_p = EnsureClosure(\delta'_p, S');$
\STATE $ms_3 = \{s_0 | \  \big ( \exists s_1, s_2:: (s_0, s_1) \in \delta_e \ \wedge (s_0, s_2) \in \delta_p \ \big ) \wedge$ $
                      \big ( \nexists s_3: (s_0, s_3) \in \delta'_p  \big ) \};$
\STATE $ms_4 = \{ s_0| \exists s_1 :: (s_1 \in ms_3 \wedge (s_0, s_1) \in \delta_e) \}$
\STATE \label{mas:removeMs4} $S' = RemoveDeadlock (S' - ms_4, \delta_p , \delta_e)$
\UNTIL ($S''' = S'$) \label{mas:ms4End}
\UNTIL \label{mas:outerLoopEnd} $(S'' = S')$ 
\RETURN $(\delta'_p, S')$ ;

\end{algorithmic}
}
\end{algorithm}

Our second variation is related to the invariant of the revised program, $S'$, and the invariant of the original program, $S$. In case of adding stabilization, we considered $S'=S$ whereas in case of adding failsafe and masking fault-tolerance, we considered $S' \subseteq S$. 

{\bf Changes to add stabilization and fault-tolerance with eventually fair environment. } \
No changes are required to Algorithms \ref{ALG:K2} or \ref{alg:GeneralK} even if environment is eventually fair. This is due to the fact that these algorithms construct programs that provide recovery from {\em any} state, i.e., they will provide recovery from the state reached after the point when fairness is restored. For Algorithms \ref{ALG:FAILSAFE} and \ref{ALG:MASKING}, we should change the input $f$ to include $\delta_e\cup \delta_f$. The resulting algorithm will ensure that the generated program will allow unfair execution of the program in initial states. However, appropriate fault-tolerance will be provided when the fairness is restored.

{\bf Changes to add stabilization and fault-tolerance with multiple consecutive environment actions. } \
If environment actions can execute consecutively, we can change input $\delta_e$ to be its transitive closure. In other words, if $(s_0, s_1)$ and $(s_1, s_2)$ are transitions in $\delta_e$, we add $(s_0, s_2)$ to $\delta_e$. With this change, the constructed program will provide the appropriate level of fault-tolerance (stabilizing, failsafe or masking) even if environment transitions can execute consecutively. 

{\bf Changes to add stabilization and fault-tolerance based on relation between $S'$ (invariant of the fault-tolerant program) and $S$ (invariant of the fault-intolerant program)}
No changes are required to Algorithms \ref{ALG:K2} or \ref{alg:GeneralK} even if we change the problem statement to allow $S' \subseteq S$ without affecting soundness or completeness. Regarding soundness, obsrve that the program generated by these algorithms ensure $S' = S$. Hence, they trivially satisfy $S' \subseteq S$. Regarding completeness, the intuition is that if it were impossible to recover to states in $S$ then it is impossible to recover to states that are a subset of $S$. 
Regarding Algorithms \ref{ALG:FAILSAFE} and \ref{ALG:MASKING}, if $S'$ is required to be equal to $S$ then they need to be modified as follows: In these algorithms if any state $S$ is removed (due to it being in $ms_2$, deadlocks, etc.) then they should declare failure.

\section{Related Work}
\label{sec:related}

This paper focuses on addition of fault-tolerance properties in the presence of unchangeable environment actions. This problem is an instance of model repair where some existing model/program is repaired to add new properties such as safety, liveness, fault-tolerance, etc. 
Model repair with respect to CTL properties was first considered in~\cite{begl99}, and abstraction techniques for the same are presented in~\cite{cbsk12}. In \cite{jgb05}, authors focus on the theory of model repair for memoryless LTL properties in a game-theoretic fashion; i.e., a repaired model is obtained by synthesizing a winning strategy for a 2-player game.
Previously  \cite{bek09}, authors have considered the problem of model repair for UNITY specifications \cite{cm88}. These results identify complexity results for adding properties such as invariant properties, leads-to properties etc. 
Repair of probabilistic  algorithms has also been considered in the literature \cite{sde08}

The problem of adding fault-tolerance to an existing program has been discussed in the {\em absence of environment actions}. 
This work includes work on controller synthesis \cite{rw89,gr09,cl98}. 
A tool for automated addition of fault-tolerance to distributed programs is presented in \cite{bka12}. This work utilizes BDD based techniques to enable synthesis of programs with state space exceeding $10^{100}$. However, this work does not include the notion of environment actions that cannot be removed. Hence, applying it in contexts where some processes/components cannot be changed will result in unacceptable solutions. At the same time, we anticipate that the BDD-based techniques considered in this work will be especially valuable to improve the performance of algorithms presented in this paper.

The work on game theory \cite{ pr89a,Jobstm05, t95} has focused on the problem of repair with 2-player game where the actions of the second player are not changed. However, this work does not address the issue of fault-tolerance. Also, the role of the environment in our work is more general than that in \cite{pr89a,pr89b,Jobstm05,t95}. Specifically, in the work on game theory, it is assumed that the players play in an alternating manner. By contrast, we consider more general interaction with the environment.

In \cite{blk11}, authors have presented an algorithm for adding recovery to component based models. They consider the problem where we cannot add to the interface of a physical component. However, it does not consider the issue of unchangeable actions of them considered in this work.

\section{Application for Distributed and Cyber-Physical Systems}
\label{sec:appli}

We considered the problem of model repair for systems with unchangeable environment actions. By instantiating these environment actions according to the system under consideration, this work can be used in several contexts. We briefly outline how this work can be used in the context of distributed systems and cyber-physical systems.

One instance of systems with unchangeable actions is distributed programs consisting of several processes. Consider such a collaborative distributed program where some components are developed in house and some are third party components. It is anticipated that we are not allowed to change  third party programs during repair. In that case, we can model the actions of those processes as unchangeable environment actions, and use algorithms provided in this paper to add stabilization/fault-tolerance. Our work is directly useful in high atomicity contexts where processes can view the state of all components but can modify only their own. In low atomicity contexts where processes have private memory that cannot be read by others, we need to introduce new restrictions. Specifically, in this context, we need to consider the issue of grouping \cite{bka12} where adding or removing a transition requires one to add or remove groups of transitions. In particular, if two states $s_0$ and $s_0'$ differ only in terms of private variables of another process then including a transition from $s_0$ requires us to add a transition from $s_0'$. Extending the algorithms in this context is beyond the scope of this paper.

Another instance in this context is a cyber-physical system. Intuitively, a CPS consists of computational components and physical components. One typical constraint in repairing these systems to satisfy new requirements is that physical components cannot be modified due to complexity, cost, or their reliance on natural laws about physics, chemistry etc. In other words, to repair a CPS model, we may not be allowed to add/remove actions which model physical aspects of the system. 
Therefore, using the approach proposed here, we can model such physical actions as unchangeable environment actions. After modeling the CPS, we can utilize the algorithms provided in this paper to add stabilization/fault-tolerance automatically, and be sure that the stabilizing/fault-tolerant models found by the algorithms do not require any change to physical components.

\section{Conclusion}
\label{sec:concl}

In this paper, we focused on the problem of adding fault-tolerance to an existing program which consists of some actions that are unchangeable. These unchangeable actions arise due to interaction with the environment, inability to change parts of the existing program, constraints on physical components in a cyber-physical system, and so on.

We presented algorithms for adding stabilization, failsafe fault-tolerance, masking fault-tolerance and nonmasking fault-tolerance. These algorithms are sound and complete and run in polynomial time (in the state space). This was unexpected in part because environment actions can play both an assistive and disruptive role. The algorithm for adding failsafe fault-tolerance was obtained by an extension of previous algorithm \cite{alisbookchapter} that added failsafe fault-tolerance without environment actions. However, the algorithms for masking and stabilizing fault-tolerance required a significantly different approach in the presence of environment actions.

We considered the cases where (1) all fault-free behaviors are preserved in the fault-tolerant program, or (2) only a nonempty subset of fault-free behaviors are preserved in the fault-tolerant program. We also considered the cases where (1) environment actions can execute with any frequency for an initial duration and (2) environment actions can execute more frequently than programs. In all these cases, we demonstrated that our algorithm can be extended while preserving soundness and completeness. Finally, as discussed in Section \ref{sec:appli}, these algorithms are especially  useful for repairing CPSs as well as repairing distributed systems where only a subset of processes are repairable. 

\bibliographystyle{unsrt}
\bibliography{bib}
\newpage
\appendix
\section{Proofs of Algorithm 
\ref{ALG:K2}}
\label{sec:2stab}

Based on the notion of fairness for program actions, we introduce the notion of whether an environment transition can be executed in a given computation prefix. Environment action can execute in a computation prefix if an environment action exists in the last state of the prefix and either (1) program cannot execute in the last state of the prefix or (2) the program has already executed $k\!-\!1$ steps. Thus, 

\begin{definition} [environment-enabled]
In any prefix $\sigma = \langle s_0, s_1,  \ldots, s_i\rangle$ of $p[]_k\delta_e$, $s_i$ is an environment-enabled state \textit{iff} \newline
$(\exists s :: (s_i, s) \in \delta_e) \wedge \newline
\Big ( \big(\nexists (s_i, s'):: (s_i, s') \in \delta_p \big ) \vee \newline \big (  \nexists j: j > i - k : (s_j, s_{j+1}) \in \delta_e \big) \Big )$.    
\end{definition}

\begin{lemma}
\label{self:RtoS}
Every computation of $p'[]_k\delta_e$ that starts from a state in $R$, contains a state in $S'$.
\end{lemma}

\begin{proof}
We prove this by induction. 

\noindent \textbf{Base case}: $R = S$. The statement is satisfied trivially. 

\noindent \textbf{Induction step}: A state $s_0$ is added to $R$ in two cases :\\ \\
\textbf{Case 1} $ (\nexists s_2 : s_2 \in \neg (R \cup R_p) :  (s_0, s_2) \in \delta_e) \wedge (\exists s_1 : s_1 \in (R \cup R_p) : (s_0, s_1) \in \delta_e)$ \\
Since there is no $s_2$ in $\neg (R \cup R_p)$ such that $(s_0, s_2) \in \delta_e$, for every $(s_0, s_1) \in \delta_e$, $s_1$ is in $R \cup R_p$. In addition, we know that there is at least one $s_1$ in $R \cup R_p$ such that $(s_0, s_1) \in \delta_e$.
If $s_1$ is in $R_p$, there is a program transition from $s_1$ to a state in $R$. As $(s_0, s_1) \in \delta_e$, because of fairness assumption, the program can occur, and reach $R$. Thus, every computation starting from $s_0$ has a state in $R$ (in the previous iteration). Hence, every computation starting from $s_0$ has a state in $S$.\\ \\
\textbf{Case 2}  $\nexists s_2 :s_2 \in \neg (R \cup Rp) :  (s_0, s_2) \in \delta_e \wedge s_0 \in R_p$\\ 
Since there is no $s_2$ in $ \neg (R \cup Rp)$ such that $(s_0, s_2) \in \delta_e$, for every $(s_0, s_1) \in \delta_e$, $s_1$ is either in $R$ or $R_p$. In addition, we know that there is at least one state $s_1$ in $R \cup R_p$ such that $(s_0, s_1) \in \delta_p$. In any computation of $p'[]_k\delta_e$ starting from $s_0$ if $(s_0, s_1) \in \delta'_p$ then $s_1 \in R$. If $(s_0, s_1) \in \delta_e$ then $s_1 \in R \cup R_p$. If $s_1$ is in $R_p$, there is a program transition from $s_1$ to a state in $R$. As $(s_0, s_1) \in \delta_e$, because of fairness assumption program can reach $R$. Thus, every computation starting from $s_0$ has a state in $R$. Hence, every computation starting from $s_0$ has a state in $S$.

\end{proof}

\begin{theorem}
\label{sk2soundness}
Algorithm~\ref{ALG:K2} is sound.
\end{theorem}

\begin{proof}
At the beginning of the algorithm $\delta'_p = \delta_p | S$ and all other transitions added to $\delta'_p$ in the rest of the algorithm starts outside $S$, so $p' | S = p |S$.
Finally, the convergence condition is satisfied based on Lemma~\ref{self:RtoS} and the fact that $R$ includes all states. 
\end{proof}

Now, we focus on showing that Algorithm~\ref{ALG:K2} is complete, i.e., if there is a solution that satisfies the problem statement for adding stabilization, Algorithm \ref{ALG:K2} finds one. The proof of completeness is based on the analysis of states that are not in $R$ upon termination. 

\begin{observation}
\label{notR1}
For any $s_0$ such that $s_0 \notin R$, we have \newline 
$\exists s_2 : s_2 \in \neg (R \cup Rp):  (s_0, s_2) \in \delta_e$, or \newline
$\nexists s_1 : s_1 \in (R \cup R_p): (s_0, s_1) \in \delta_e   \wedge  s_0 \notin R_p$.
\end{observation}

\begin{observation}
\label{obs:inRRp}
For any $s_0$ such that $s_0 \notin R$ and $\exists s_1 :: (s_0, s_1) \in \delta_e$, we have
$\exists s_2 :s_2 \in \neg (R \cup Rp) :  (s_0, s_2) \in \delta_e$.
\end{observation}

\begin{lemma}
\label{lem:notRNotRecovery}
Let $\delta''_p$ be any program such that $\delta''_p \cap \delta_b = \phi$. Let $s_j$ be any state in $\neg (R \cup R_p)$. Then, either $s_j$ is a deadlock state in $\delta''_p \cup \delta_e$, or for every $p''[]_k\delta_e$ prefix $\alpha = \langle..., s_{j-1}, s_j\rangle$, there exists suffix $\beta = \langle s_{j+1}, s_{j+2}, \ldots\rangle$, such that $\alpha\beta$ is a $p''[]_k\delta_e$ computation, and one of two conditions below is correct: 
\begin{enumerate}
\item $s_{j+1} \in \neg (R \cup R_p)$
\item $s_{j+1} \in R_p-R \wedge s_{j+2} \in \neg(R \cup R_p)$
\end{enumerate}

%
\end{lemma}

\begin{proof} There are two cases for $s_j$:\\ \\ 
\textbf{Case 1} If $s_j$ is environment-enabled\\
Based on the Observation~\ref{obs:inRRp} there should exist $s'' \in \neg (R \cup R_p)$ such that $(s_j, s'') \in \delta_e$. We set $s_{j+1} = s''$.\\ \\ 
\textbf{Case 2} If $s_j$ is not environment-enabled \\
In this case $(s_j, s_{j+1}) \in \delta_p$, and as $s_j \in \neg (R \cup R_p)$, $s_{j+1} \in \neg R$. (otherwise $s_j$ would be in $R_p$). There are two sub-cases for this case: \\ \\
\textbf{Case 2.1} $s_{j+1} \in \neg R_p$\\
In this case $s_{j+1} \in \neg(R \cup R_p)$. \\ \\  
\textbf{Case 2.2} $s_{j+1} \in  R_p$\\
 As $s_{j+1} \in \neg R \cap R_p$, according to Observation~\ref{notR1}, we have $\exists s_2 : s_2 \in \neg (R \cup Rp): (s_{j+1}, s_2) \in \delta_e$.
As $(s_j, s_{j+1}) \in \delta_p$, even with fairness $(s_{j+1}, s_2)$ can occur. Therefore we set $s_{j+2} = s_2$, i.e., $s_{j+2} \in \neg (R \cup R_p)$.
\end{proof}

\begin{corollary}
\label{notRNotRecovery}
Let $\delta''_p$ be any program such that $\delta''_p \cap \delta_b = \phi$. Let $s_j$ be any state in $\neg (R \cup R_p)$. Then for every $p''[]_k\delta_e$ prefix $\alpha = \langle..., s_{j-1}, s_j\rangle$, there exists suffix $\beta = \langle s_{j+1}, s_{j+2}, \ldots\rangle$, such that $\alpha\beta$ is a $p''[]_k\delta_e$ computation, and $\forall i: i \geq j : s_i \in \neg R$ (i.e. $\neg S$).

\end{corollary}

\begin{theorem}
\label{sk2completeness}
Algorithm~\ref{ALG:K2} is complete.
\end{theorem}

\begin{proof}
Algorithm~\ref{ALG:K2} returns Not-possible only when, at the end of loop there exists a state $s_0$ such that $s_0 \notin R$.
When $s_0 \notin R$, according to Observation~\ref{notR1} we have two cases as follows: \\ \\ 
\textbf{Case 1} $\exists s_2 : s_2 \in \neg (R \cup Rp):  (s_0, s_2) \in \delta_e$ \\
As there exists an environment action to state $s_2$ in $\neg (R \cup R_p)$, starting from $s_0$ there is a computation that next step is in $\neg (R \cup R_p)$. Note that, when a computation starts from $s_0$, even with fairness assumption $(s_0, s_2) \in \delta_e$ can occur.
Based on Corollary~\ref{notRNotRecovery}, for every $
\delta''_p$ such that $\delta''_p \cap \delta_b = \phi$, starting from $s_0$, there is a computation such that every state is in $\neg R$.\\ \\
\textbf{Case 2} $\nexists s_1 : s_1 \in (R \cup R_p): (s_0, s_1) \in \delta_e   \wedge   s_0 \notin R_p$\\
Based on Corollary~\ref{notRNotRecovery}, starting from $s_0 \in \neg (R \cup R_p)$, there is a computation such that every state is in $\neg R$. Therefor for every $\delta''_p$ such that $\delta''_p \cap \delta_b = \phi$, there is a computation starting from $s_0$ such that all states are outside $R$ (i.e outside $s$). Thus, it it impossible to have any stabilizing revision for the program.

\end{proof}

\begin{theorem}
\label{thm:alg1P}
Algorithm \ref{ALG:K2} is polynomial (in the state space of $p$)
\end{theorem}

\begin{proof}
The proof follows from the fact that each statement in Algorithm \ref{ALG:K2} is executed in polynomial time and the number of iterations are also polynomial, as in each iteration at least one state is added to $R$. 
\end{proof}

Proof of Theorem~\ref{thm:alg1} is resulted by Theorem ~\ref{sk2soundness}, \ref{sk2completeness}, and \ref{thm:alg1P}.

\section{Addition of Safe Stabilization for any k}
\label{sefAdditionG}

In this section, we present a general algorithm for addition problem defined in Section~\ref{sefAddition}. Algorithm~\ref{ALG:K2} generates a program that is stabilizing when $k\!=\!2$. Hence, the generated programs are stabilizing even with a higher $k$ value. However, Algorithm~\ref{ALG:K2} will fail to find a program if addition of stabilization requires an higher value of $k$. 
Algorithm~\ref{alg:GeneralK} is complete for any $k>1$.

In this algorithm, state predicate $R$ is the set of states such that every computation starting from them has a state in $S$. At the beginning this is equal to $S$.
In each iteration, the value of $Rank$ for each state shows the number of program transitions needed to reach a state in $R$. At the beginning, $Rank$ of all states in $R$ is equal to 0 and all the other states have $Rank$ equal to $\infty$. In each iteration, repeat loop on Line~\ref{GeneralK:insideBeign}-\ref{GeneralK:insideEnd} compute smallest $Rank$ possible for each state outside $R$, and change program transitions to reach $R$ using minimum number of program transition.

In for loop on Line~\ref{GeneralK:forBegin}-\ref{GeneralK:forEnd}, we add new states to $R$. We add a state to $R$ whenever every computation starting from that state has a state in $S$. A state $s_0$ can be added to $R$, only when there is no environment transition starting form $s_0$ going to a state with $Rank \geq k$. In addition to this condition, there should be one way for $s_0$ to reach $R$. Therefore, there should be at least one environment transition to a state with $Rank < k$, or $Rank$ of $s_0$ is 1 which means from $s_0$ we can reach a state in $R$ with only one program transition. Just like Algorithm~\ref{ALG:K2}, this algorithm terminates if no state can be added to $R$ in the last iteration. 
%
 At the end of the algorithm, if there a state outside $R$, we declare that there is no safe stabilizing revision for the original program. Otherwise, the algorithm returns the set of transitions of the revised stabilizing program, $\delta'_p$.

\begin{algorithm}
\caption{Addition of safe stabilization for any k}
\label{alg:GeneralK}
{
\begin{algorithmic} [1]
\INPUT $S_p, \delta_p, \delta_e, S, \delta_b$, and $k$
\OUTPUT $\delta'_p$ or Not-possilbe
\STATE $\delta'_p := (\delta_p | S);$
\STATE $R = S;$
\STATE $\forall s :s \in R :  Rank.s = 0;$
\STATE $\forall s  : s \in R : Rank.s = \infty;$

\REPEAT \label{GeneralK:outsideBegin}
\STATE $R' = R;$
\REPEAT \label{GeneralK:insideBeign}

\STATE $\delta''_p = \delta'_p;$
\IF {$\exists s_0 : s_0 \notin R: (\exists s_1 : Rank.s_1 + 1 < Rank.s_0 : (s_0, s_1) \notin \delta_b)$}
\STATE $\delta'_p = \delta'_p - \{(s_0, s) | \ (s_0, s) \in  \delta'_p\}$;
\STATE $\delta'_p = \delta'_p\cup \{(s_0, s_1)\}$;
\STATE $Rank.s_0 = Rank.s_1 + 1$;
\ENDIF
\UNTIL \label{GeneralK:insideEnd} $(\delta''_p = \delta'_p)$

\FOR{$each \ s_0 : s_0 \notin  R :$ \newline$(\nexists s_2: Rank.s_2 \geq k : (s_0, s_2) \in \delta_e) \wedge$ \newline $\big((\exists s_1: Rank.s_1 < k : (s_0, s_1) \in \delta_e )  \lor   (Rank.s_0=1) \big )$}  \label{GeneralK:forBegin}
\STATE $Rank.s_0 = 0;$
\STATE $R = R \cup {s_0};$	
\ENDFOR \label{GeneralK:forEnd}
\UNTIL $(R' = R)$ \label{GeneralK:outsideEnd}
\IF{$\exists s_0 : s_0 \notin  R $}
\RETURN  Not-possible;
\ELSE
\RETURN  $\delta'_p;$
\ENDIF

\end{algorithmic}
}
\end{algorithm}

Now, we provide the proofs of Algorithm  \ref{alg:GeneralK}.

\begin{lemma}
\label{selfG:RtoS}
Every computation of $p'[]_k\delta_e$ that starts from a state in $R$, contains a state in $S'$.
\end{lemma}
\begin{proof} Proof by induction:

\noindent \textbf{Base case}: $R = S$

\noindent \textbf{Induction Step}:
There are two cases that the algorithm adds $s_0$ to $R$ (i.e., set $Rank$ to 0):\\ \\ 
\textbf{Case 1} $(\nexists s_2 : Rank.s_2 \geq k  : (s_0, s_2) \in \delta_e) \wedge (\exists s_1: Rank.s_1 < k : (s_0, s_1) \in \delta_e)$ \\
In this case any environment transition starting from $s_0$ reaches state with $Rank < k$. Therefore, with less than $k$ transitions, it is possible to reach $S$. Moreover, at least one such transition exists, so reaching $S$ from this state is guaranteed. \\ \\
\textbf{Case 2} $(\nexists s_2: Rank.s_2 \geq k  : (s_0, s_2) \in \delta_e) \wedge (Rank.s_0 = 1)$ \\
In this case there is no environment action from $s_0$ to a state with $Rank \geq k$, and program can reach state with $Rank = 0$ with one step by a program transition. So, recovery to $S$ is guaranteed.
\end{proof}

\begin{theorem}
\label{thm:GKson}
Algorithm~\ref{alg:GeneralK} is sound.
\end{theorem}

\begin{proof}
Proof of this theorem is quite the same as Theorem~\ref{sk2soundness} just instead of Lemma~\ref{self:RtoS}, we should use Lemma~\ref{selfG:RtoS}.
\end{proof}

\begin{observation}
\label{kg0}
For any state $s_0$ if $Rank.s_0 > 0$ (i.e.,  $s_0 \notin R$) then
\begin{enumerate}
   \item $\exists s_2: Rank.s_2 \geq k  : (s_0, s_2) \in \delta_e$, or
   \item ($\nexists s_1: Rank.s_1 < k : (s_0, s_1) \in \delta_e   )  \wedge  Rank.s_0 > 1$
\end{enumerate}
\end{observation}

\begin{observation}
\label{eCanFoTok}
Let $p''$ be any program such that $\delta''_p \cap \delta_b = \phi$. Let $s_j$ be any state with $Rank  > 0$, and $\exists s :: (s_j,s) \in \delta_e$. Then for every $p''[]_k\delta_e$ prefix $\alpha = \langle..., s_{j-1}, s_j\rangle$, there exists suffix $\beta = \langle s_{j+1}, s_{j+2}, \ldots\rangle$, such that $\alpha\beta$ is a $p''[]_k\delta_e$ computation, and $Rank.s_{j+1} \geq k \wedge (s_j, s_{j+1}) \in \delta_e$.

\end{observation}

\begin{observation}
\label{mtom-1}
Let $p''$ be any program such that $\delta''_p \cap \delta_b = \phi$. Let $s_j$ be any state with $Rank  = m$, there does not exist state $s'$ with $Ranks < m-1$ such that $(s, s') \in \delta''_p$. 
\end{observation}

\begin{theorem}
\label{morm-1}
Let $p''$ be any program such that $\delta''_p \cap \delta_b = \phi$. Let $s_j$ be any state with $Rank  = m $. Then for every $p''[]_k\delta_e$ computation $\langle ..., s_{j-1}, s_j, s_{j+1}, \ldots\rangle$ where $\forall l : l \geq j : (s_l, s_{l+1}) \in \delta_p$, one of two following conditions is true:
%
\begin{enumerate}
\item $\exists i: i > j :Rank.s_i = m-1.$
\item $\forall i: i \geq j : Rank.s_i \geq m.$
\end{enumerate}

\end{theorem}

\begin{proof}
We proof this theorem by contradiction: \newline
Suppose it is not true. It means there is a $p''[]_k\delta_e$ computation $\langle ..., s_{j-1}, s_j, s_{j+1}, \ldots\rangle$ where $Rank.s_j = m$  and $\forall l : l \geq j : (s_l, s_{l+1}) \in \delta_p$, but \newline $\exists k : k \geq j : (\forall l : s \leq l \leq k : Rank.l \geq m) \wedge Rank.(k+1) < m - 1$. This is contradiction to Observation~\ref{mtom-1}.
\end{proof}

\begin{corollary}
\label{needs1}
Let $p''$ be any program such that $\delta''_p \cap \delta_b = \phi$. Let $s_j$ be any state with $Rank  = m $ and $s_k$ be any state with $Rank = 0$. Then for every $p''[]_k\delta_e$ prefix $\langle ..., s_{j-1}, s_j, s_{j+1}, \ldots, s_k\rangle$ where $\forall l : l \geq j : (s_l, s_{l+1}) \in \delta_p$, $\exists i: j < i < k : Rank.s_i =1$.
\end{corollary}

\begin{corollary}
\label{needm-1}
Let $p''$ be any program such that $\delta''_p \cap \delta_b = \phi$. Let $s_j$ be any state with $Rank  = m $ and $s_k$ be any state with $Rank = n$. Then for every $p''[]_k\delta_e$ prefix $\langle ..., s_{j-1}, s_j, s_{j+1}, \ldots, s_k\rangle$ where $\forall l : l \geq j : (s_l, s_{l+1}) \in \delta_p$, $ k - j \geq m-n$.
\end{corollary}

\begin{theorem}
Let $p''$ be any program such that $\delta''_p \cap \delta_b = \phi$. Let $s_j$ be any state with $Rank \geq k$. Then for every $p''[]_k\delta_e$ prefix $\alpha = \langle..., s_{j-1}, s_j\rangle$, there exists suffix $\beta = \langle s_{j+1}, s_{j+2}, \ldots\rangle$, such that $\alpha\beta$ is a $p''[]_k\delta_e$ computation, and 
one two cases below is correct: 
\begin{enumerate}
\item $\forall i: j < i : Rank.s_i > 0$
\item  $\exists i: j < i  : Rank.s_i \geq k \wedge \forall l: j < l : Rank.s_l > 0$
\end{enumerate}

\end{theorem}

\begin{proof}
%
If $\exists s :: (s_0,s) \in \delta_e$, according to Observation~\ref{eCanFoTok} there should be an environment action to state with $Rank \geq k$. Therefore, that transition can occur and program reaches a state with $Rank \geq k$.
Otherwise, if there is no environment transition starting from $s_0$, $s_0$ is either a deadlock state or has a program transition in $\delta''_p$. If it is deadlock the theorem is proved. Otherwise the same property is true in the next state $s_1$ that is reached by the program action if $Rank.s_1 > 1$. If there is no deadlock or environment transition for a sequence of states, then we have a sequence of program transitions. If $Rank$ of all states in this sequence is greater that 0, the theorem is proved. Otherwise, if there is $s_k$ with $Rank = 0$  according to Corollary~\ref{needs1} we reach state $s'$ with $Rank =1$. According to Observation~\ref{kg0} from $s'$ there is an environment transition $(s', s'')$ to a state with $Rank \geq k$. According to Corollary~\ref{needm-1}, to reach $s'$ with $Rank = 1$ from $s_j$ with $Rank \geq k$ at least $k-1$ program transitions are needed. It means that even with fairness assumption an environment transitions can occur from $s'$. Therefore, $(s', s'')$  can occur and reach $s''$ with $Rank \geq k$. 

\end{proof}

\begin{corollary}
\label{finalCG}
Let $\delta''_p$ be any program such that $\delta''_p \cap \delta_b = \phi$. Let $s_j$ be any state with $Rank \geq k$. Then for every $p''[]_k\delta_e$ prefix $\alpha = \langle..., s_{j-1}, s_j\rangle$, there exists suffix $\beta = \langle s_{j+1}, s_{j+2}, \ldots\rangle$, such that $\alpha\beta$ is a $p''[]_k\delta_e$ computation, and 
 $\forall i: i \geq j  : Rank.s_i > 0$.
 

\end{corollary}

\begin{theorem}
\label{thm:GKcom}
Algorithm~\ref{alg:GeneralK} is complete.
\end{theorem}
\begin{proof}
According to Observation~\ref{kg0}, if we have a state $s_0$ with $Rank > 0$, we have two cases:\\ \\
\textbf{Case 1} $\exists s_2: Rank.s_2 \geq k  : (s_0, s_2) \in \delta_e$ \\
As there exists an environment action to state $s_2$ with $Rank \geq k$, starting from $s_0$ there is computation that next state has $Rank \geq k$. Note that, when a computation starts from $s_0$, even with fairness assumption $(s_0, s_2) \in \delta_e$ can occur. Based on Corollary~\ref{finalCG}, starting from $s_0$, there is a computation such that every state is outside $S$. \\ \\
 \textbf{Case 2}  $(\nexists s_1: Rank.s_1 < k : (s_0, s_1) \in \delta_e )    \wedge  (Rank.s_0 > 1)$ \\
If $\exists s :: (s_0,s) \in \delta_e$, according to Observation~\ref{eCanFoTok} there should be an environment action to state with $Rank \geq k$. Therefore, that transition can occur and program reaches a state with $Rank \geq k$.
Otherwise, if there is no environment transition starting from $s_0$, $s_0$ is either a deadlock state or has a program transition in $\delta''_p$. If it is deadlock it means that there is no computation to reach $S$. Otherwise the same property is true in the next state, $s_1$, that is reached by the program action if $Rank.s_1 > 1$. If there is no deadlock or environment transition for a sequence of states, then we have a sequence of program transitions. If $Rank$ of all states in this sequence is greater that 0, it means that there is a computation that does not reach S, and the theorem is proved. Otherwise, if there is $s_k$ with $Rank = 0$  according to Corollary~\ref{needs1} we reach state $s'$ with $Rank =1$. According to Observation~\ref{kg0} from $s'$ there is an environment transition $(s', s'')$ to a state with $Rank \geq k$. According to Corollary~\ref{needm-1}, to reach $s'$ with $Rank = 1$ from $s_0$ with $Rank \geq k$  at least $k-1$ program transitions are needed. It means that even with fairness assumption an environment transitions can occur from $s'$. Therefore, $(s', s'')$  can occur and reach $s''$ with $Rank \geq k$. Then, according to Corollary~\ref{finalCG}, there is a computation starting from $s''$ (i.e., $s_0$) in which all the states are outside $S$.
\end{proof}

\begin{theorem}
\label{thm:GKP}
Algorithm \ref{alg:GeneralK} is polynomial (in the state space of $p$)
\end{theorem}

\begin{proof}
The proof follows from the fact that each statement in Algorithm \ref{alg:GeneralK} is executed in polynomial time and the number of iterations are also polynomial, as in each iteration at least one state is added to $R$.
\end{proof}

\section{Proofs for Algorithm \ref{ALG:FAILSAFE}}
\label{sec:failsafeproofs}

In this section, we prove the soundness (Theorem \ref{fail:soundness}) and completeness (Theorem \ref{fail:completeness}) of this algorithm with the help of Lemma \ref{th:C1}, Corollaries \ref{fai:noms2} and \ref{fai:nomt}, and Theorem \ref{fai:noms4}. We also prove the complexity result for this algorithm (Theorem \ref{thm:fsP}). 

The following theorem splits condition $C1$ into easily checkable conditions that assist in soundness and completeness. Specifically, this shows that condition $C1$ is satisfied iff $p'$ does not include any new states or transitions in $S'$. It also ensures that new computations are not created in $p'$ due to deadlocks (caused by removal of transitions from $p$) that may be created due to removal of transitions of $p$.  

\begin{lemma}
\label{th:C1}
The condition \textbf{\textit{C1}} in the problem definition of addition of fault-tolerance is satisfied for $k=2$ \textit{iff} conditions below are satisfied:
\begin{enumerate}
\item	$\delta'_p \cup \delta_e$ is closed in $S'$
\item	$S' \subseteq S$
\item	$\delta'_p|S' \subseteq \delta_p|S$
\item	$\forall s_1 : \big ( \exists s_0, s_2, s_3 : s_0 \in S' \wedge (s_0,s_1), (s_1,s_2) \in \delta_e \wedge (s_1,s_3) \in \delta_p  \big ): \big ( \exists s_4 :: (s_1,s_4) \in \delta'_p \big )$
\item $\exists (s_0, s_1) :: (s_0, s_1) \in (\delta_p \cup \delta_e) \Rightarrow  \exists (s_0, s_2) :: (s_0, s_2) \in (\delta'_p \cup \delta_e)$
\end{enumerate}

\begin{proof}

$(\Rightarrow)$ If any of the five conditions are violated then we can easily create a new computation of $p'[]_2\delta_e$ that is not a computation of $p[]_2\delta_e$ thereby violating \textit{C1}. 

($\Leftarrow$)
We show by induction that if the five conditions of theorem are satisfied, then every prefix $\sigma = \langle s_0, s_1, \ldots, s_i \rangle$ of $p'[]_2\delta_e$ that starts from a state in $S'$ is a prefix of $p[]_2\delta_e$ which starts in $S$.

As $S' \subseteq S$, we know that $s_0 \in S$, then for every $i > 0$ we have cases below:\\ \\
\textbf{Case 1} $(s_i, s_{i+1}) \in \delta'_p$ \\
Since $p' \cup \delta_e$ is closed in $S'$, we know that $s_{i+1} \in S'$. As $\delta'_p| S' \subseteq \delta_p|S'$, we have $(s_i, s_{i+1}) \in \delta_p$. Therefore, if $\langle s_0, s_1, \ldots, s_i \rangle$ is a prefix  of $p []_2 \delta_e$, then $\langle s_0, s_1, \ldots, s_i, s_{i+1} \rangle$ is a prefix of $p []_2 \delta_e$. \\ \\
\textbf{Case 2} $(s_i, s_{i+1}) \in \delta_e$ : \\
Two sub-cases are possible below this case: \\ \\ 
\textbf{Case 2.1}
 $(s_{i-1}, s_i) \in \delta'_p$ \\ 
In this case, as we have reached $s_i$ by a program transition, even with fairness assumption, $(s_i, s_{i+1})$ can occur in $p []_2 \delta_e$, Therefore, if $\langle s_0, s_1, \ldots, s_i \rangle$ is a prefix  of $p []_2 \delta_e$, then $\langle s_0, s_1, \ldots, s_i, s_{i+1} \rangle$ is a prefix of $p []_2 \delta_e$. \\ \\ 
\textbf{Case 2.2}$(s_{i-1}, s_i) \in \delta_e$ \\
In this case, there should not be exist $s'$ such that $(s_i, s') \in \delta'_p$(otherwise, because of fairness $(s_i, s_{i+1})$ cannot be in any prefix of $p' []_2 \delta_e$). Then, according to condition 3, there should not exist state $s''$ such that $(s_i, s'') \in \delta_p$. If not, assumption 3 of the theorem is violated. As there is no such $s''$, even with fairness assumption $(s_i, s_{i+1})$ can occur in $p []_2 \delta_e$. Therefore, if $\langle s_0, s_1, \ldots, s_i \rangle$ is a prefix  of $p []_2 \delta_e$, then $\langle s_0, s_1, \ldots, s_i, s_{i+1} \rangle$ is a prefix of $p []_2 \delta_e$. \\ \\ 
\textbf{Case 3} $s_i$ is deadlock in $\delta'_p \cup \delta_e$ \\
From condition 5, $s_i$ is deadlock in $\delta_p \cup \delta_e$ as well, so $\langle s_0, s_1, \ldots, s_i \rangle$ is a computation of $p []_2 \delta_e$.

As every prefix of $p'[]_2\delta_e$ that start from a state in $S'$ is a prefix of $p[]_2\delta_e$ which starts in $S$, \textit{C1} is satisfied.

\end{proof}

\end{lemma}

\begin{theorem}
\label{fail:soundness}
Algorithm~\ref{ALG:FAILSAFE} is sound.
\end{theorem}

\begin{proof}
To show the soundness of our algorithm, we need to show that the three conditions of the addition problem are satisfied.

\textbf{\textit{C1}}:
Consider a computation $c$ of $p'[]_2 \delta_e$ that starts from a state in $S'$. By construction, $c$ starts from a state in $S$, and $\delta'_p|S'$ is  a subset of $\delta_p|S'$. In addition, $\delta'_p \cup \delta_e$ is closed in $S'$. Therefore, the first three requirements of Lemma~\ref{th:C1} are satisfied. Now, we show the forth and fifth requirements of Lemma~\ref{th:C1} are satisfied, as well.

Regarding fourth requirement of Lemma~\ref{th:C1}, suppose that there exists $s_1$ in $S'$ such that  $\big ( \exists s_0, s_2, s_3 : s_0 \in S' \wedge (s_0,s_1), (s_1,s_2) \in \delta_e \wedge (s_1,s_3) \in \delta_p  \big )$ but $\nexists s_4 :: (s_1,s_4) \in \delta'_p$. From $\exists s_2, s_3 :: (s_1,s_2) \in \delta_e \wedge (s_1,s_3) \in \delta_p$ and $\nexists s_4 :: (s_1,s_4) \in \delta'_p$  we can conclude that $s_1$ is in $ms_3$. Then, from $(s_0,s_1) \in \delta_e$ , we know $s_0$ is in $ms_4$, which is contradiction as $s_0$ is in $S$.

Finally, the fifth requirement is satisfied based on our approach for dealing with deadlock states in Assumption~\ref{assumption:infinite}. 
Hence, \textit{C1} holds.

\textit{\textbf{C2}}:
From \textit{C1}, and the assumption that $p[]_2\delta_e$ refines $spec$ from $S$, $p'[]_2\delta_e$  refines $spec$ from $S'$.

Let $spec$ be $\langle Sf, Lv \rangle$. Consider prefix $c$ of $p'[]_2 \delta_e []f$  such that $c$ starts from a state in $S'$. If $c$ does not refine $Sf$ then there exists a prefix of $c$, say $\langle s_0, s_1,\ldots, s_n \rangle$, such that it has a transition in $\delta_b $. Wlog, let $\langle s_0, s_1, \cdots, s_n\rangle$ be the smallest such prefix. It follows that $(s_{n-1}, s_n) \in \delta_b$, hence, $(s_{n-1}, s_n) \in mt$. By construction, $p'$ does not contain any transition in $mt$. Thus, $(s_{n-1}, s_n)$ is a transition of $f$ or $\delta_e$. If it is in $f$ then $s_{n-1} \in ms_1$ (i.e., $s_{n-1} \in ms_2$). If it is in $\delta_e$ then $s_{n-1} \in ms_2$. Therefore, in both cases, $s_{n-1} \in ms_2$, and $(s_{n-2}, s_{n-1}) \in mt$. Again, by construction we know that $\delta'_p$ does not contain any transition in $mt$, so $(s_{n-2}, s_{n-1})$ is either in $f$ or $\delta_e$. If it is in $f$ then $s_{n-2} \in ms_1$ (i.e., $s_{n-2} \in ms_2$). If it is in $\delta_e$ two cases are possible:

1) $(s_{n-1}, s_n) \in f$. In this case, as stated before, $s_{n-1} \in ms_1$, so $s_{n-2} \in ms_2$. 

2) $(s_{n-1}, s_n) \in \delta_e$. In this case, all transitions starting from $s_{n-1}$ should be in $mt$. If this is not the case then this implies that there exists a state $s$ such that $(s_{n-1}, s)$ is not in $mt$ and we would have added it to $\delta'_p$ by Line~\ref{fail:line:finaldelta}. 

Since all transitions from $s_{n-1}$ are in $mt$, $s_{n-1}$ is in $ms_1$ (by Line \ref{fai:addtoMs1}). Hence, $s_{n-2}$ is in $ms_2$.

Continuing this argument further leads to the conclusion that $s_0 \in ms_2$. This is a contradiction. Thus, any prefix of $p'[]_2 \delta_e []f$ refines $Sf$. Thus, \textit{C2} holds.

\textit{\textbf{C3}}:
Any $(s_0, s_1) \in \delta_r$, is in $mt$. By construction, $\delta'_p$ does not have any transition in $mt$. Hence, \textit{C3} holds.
\end{proof}

Now, we focus on showing that Algorithm \ref{ALG:FAILSAFE} is complete, i.e., if there is a solution that satisfies the problem statement for adding failsafe fault-tolerance, Algorithm \ref{ALG:FAILSAFE} finds one. The proof of completeness is based on the analysis of states that were removed from $S$. 

\begin{observation}
\label{failsafe:obs1}
For every state $s_0$ in $ms_2$ one of three cases below is true:
\begin{enumerate}
   \item $s_0 \in ms_1$
   \item $\exists s_1:: (s_0, s_1) \in \delta_e \wedge s_1 \in ms_1$
   \item  $\exists s_1:: (s_0, s_1) \in \delta_e \cap \delta_b $
\end{enumerate}

\end{observation}

\begin{theorem}
\label{fai:reachms1}
Let $p''$ be any program that solves the problem of adding failsafe fault-tolerance. For every prefix $\alpha = \langle \ldots, s_{j-1}, s_j\rangle$ of $p''[]_2\delta_e$ where $s_j \in ms_2$ and $(s_{j-1}, s_j) \in f \cup \delta''_p$, there exists a suffix $\beta =\langle s_{j+1}, s_{j+2}, \ldots\rangle$ such that $\alpha\beta$ is a computation of $p''[]_2\delta_e$ and $ \exists i : i \geq j : (s_i \in ms_1) \vee ((s_i, s_{i+1}) \in \delta_b)$.
\end{theorem}

\begin{proof}
According to Observation~\ref{failsafe:obs1}, $s_j$ is either in $ms_1$, can reach a state in $ms_1$ by an environment action, or has an environment action in $\delta_b$. If $s_j \in ms_1$, theorem is proved. If $s_j$ is not in $ms_1$, there exists an environment action $e$ which is either in $\delta_b$, or reaches a state in $ms_1$. As we have reach $s_1$ by a transition in $\delta''_p \cup f$, even with fairness assumption, $e$ can be executed. Thus, either a transition in$\delta_b$ occur or a state in $ms_1$ is reached. 
\end{proof}

\begin{theorem}
\label{fai:ms1}
Let $p''$ be any program that solves the problem of adding failsafe fault-tolerance. For every prefix $\alpha = \langle \ldots, s_{j-1}, s_j\rangle$ of $p''[]_2\delta_e$ where $s_j \in ms_1$, there exists a suffix $\beta =\langle s_{j+1}, s_{j+2}, \ldots\rangle$ such that $\alpha\beta$ is a computation of $p''[]_2\delta_e$ and $\exists i : i \geq j: (s_i, s_{i+1}) \in \delta_b$.
\end{theorem}

\begin{proof}
We prove this inductively based on when states are added to $ms_1$

\noindent \textbf{Base case} $ms_1 = \{s_0 | (s_0, s_1) \in  f \cap \delta_b\}$ 

Since fault transitions can execute in any state, the theorem is satisfied by construction. 

\noindent \textbf{Induction step} A state $s_0$ is added into $ms_1$ in three cases:\\ \\
\textbf{Case 1} $\exists s_1: s_1 \in  ms_2: (s_0 , s_1) \in f$ \\ 
In this case according to Theorem~\ref{fai:reachms1}, a transition in $\delta_b$ may occur, or a state in $ms_1$ can be reached. Hence, according to induction hypothesis a transition in $\delta_b$ can occur in both cases. \\ \\
\textbf{Case 2} $\exists s_1:: (s_1 \in ms_1  \wedge (s_0, s_1) \in \delta_e) \wedge (\forall s_2:: (s_0, s_2) \in mt)$ \\
In this case, if according to fairness $(s_0, s_1)$ can occur, state $s_1 \in ms_1$ can be reached by $(s_0, s_1)$, and according to induction hypothesis safety may be violated. However, if $(s_0, s_1)$ cannot occur, some other transition in $\delta''_p \cup f$ should occur, but we know such transition should be in $mt$ and reaches a state in $ms_2$. Thus, according to Theorem~\ref{fai:reachms1} either a transition in $\delta_b$ can occur, or a state in $ms_1$ can be reached. Hence, according to induction hypothesis a transition in $\delta_b$ can occur in both cases.\\ \\ 
\textbf{Case 3} $((s_0, s_1) \in \delta_e \cap \delta_b) \wedge (\forall s_2:: (s_0, s_2) \in mt)$ \\
In this case, if according to fairness $(s_0, s_1)$ can occur, by its occurrence safety is violated. However, if $(s_0, s_1)$ cannot occur, some other transition in $\delta''_p \cup f$ should occur, but we know such transition should be in $mt$ and reaches a state in $ms_2$. Thus, according to Theorem~\ref{fai:reachms1} either a transition in $\delta_b$ can occur, or a state in $ms_1$ can be reached. Hence, according to induction hypothesis a transition in $\delta_b$ can occur in both cases.
\end{proof}
According to Theorem~\ref{fai:reachms1} and Therorem~\ref{fai:ms1} we have the following two corollaries.
\begin{corollary}
\label{fai:noms2}

Let $p''$ be any program that solves the problem of adding failsafe fault-tolerance and let $S''$ be its invariant. Then, $S'' \cap ms_2 = \phi$.
\end{corollary}

\begin{corollary}
\label{fai:nomt}
Let $p''$ be any program that solves the problem of adding failsafe fault-tolerance and let $S''$ be its invariant. $p''|S''$ cannot have any transition in $mt$. 

\end{corollary}

\begin{theorem}
\label{fai:noms4}
Let $p''$ be any program that solves the problem of adding failsafe fault-tolerance, and let $S''$ be its invariant. Then $S''$ cannot include any state in set $ms_4$ in any iteration of loop on Lines~\ref{fail:removems4Begin} - \ref{fail:removems4End}. 
\end{theorem}
\begin{proof}
We show that if $s_0$ is in $S''$, and $s_0$ is in $ms_4$ in any iteration of loop on Lines~\ref{fail:removems4Begin} - \ref{fail:removems4End}, then there is sequence $\sigma = \langle s_0, s_1, \ldots \rangle$ such that $\sigma$ is $p''[]_2\delta_e$ computation, but it is not a $p[]_2\delta_e$ computation.

Suppose $s_0$ is in $ms_4$ in the first iteration of the loop. then there is a state, $s \in ms_3$ such that $(s_0, s) \in \delta_e$. Thus, $\big ( \exists s_1, s_2:: (s, s_1) \in \delta_e \ \wedge (s, s_2) \in \delta_p \ \big ) \wedge$ $
                      \big ( \nexists s_3:: (s, s_3) \in \delta'_p  \big )$. Sine $p''$ solves the problem, $S''$ is closed in $p''[]_2\delta_e$. Hence, $s$ is in $S''$, as well. 
In the first iteration  $\delta''_p|S'' \subseteq \delta'_p$, because according to Corollary~\ref{fai:noms2} and Corollary~\ref{fai:nomt} $\delta''_p$ cannot have a transition in $mt$, and $S-ms_2$ should be closed in $\delta''_p \cup \delta_e$. Therefore, $\big ( \exists s_1, s_2:: (s, s_1) \in \delta_e \ \wedge (s, s_2) \in \delta_p \ \big ) \wedge$ $
                      \big ( \nexists s_3:: (s, s_3) \in \delta''_p  \big )$.

Now, observe that $\langle s_0, s, s_1, \cdots \rangle$ is $p''[]_2\delta_e$ computation, but it is not a $p[]_2\delta_e$ computation, as because of fairness $(s,s_1)$ cannot occur when there exist $(s, s_2) \in \delta_p$. Since $p''$ solves the addition problem, it cannot have any state in $ms_4$ in the first iteration. Therefore, as $S''$ should be closed in $p''[]_2\delta_e$, all transition to states in $ms_4$ in the first iteration should be removed from $\delta''_p$. Thus, $\delta''_p|S'' \subseteq \delta'_p$ in the second iteration as well, and with same argument we can show that if $s_0 \in S''$ is in $ms_4$ in any iteration, then there is sequence $\sigma = \langle s_0, s_1, \cdots \rangle$ such that $\sigma$ is $p''[]_2\delta_e$ computation, but it is not a $p[]_2\delta_e$ computation.
\end{proof}

\begin{theorem}
\label{fail:completeness}
Algorithm~\ref{ALG:FAILSAFE} is complete.
\end{theorem}
\begin{proof}
Let program $p''$ and predicate $S''$ solve transformation problem. $S''$ should satisfy following requirements: 
\begin{enumerate}
\item $S'' \cap ms_2 = \phi$
\item $S''$ does not include any state in set $ms_4$ in any iteration of loop on Lines~\ref{fail:removems4Begin} - \ref{fail:removems4End}
\item $\nexists s_0 : s_0 \in S'' : (\exists s_1 : s_1 \notin S'' : (s_0, s_1) \in \delta_e)$
\end{enumerate} 

The first requirement is according to Corollary~\ref{fai:noms2}. The second requirement is according to Theorem~\ref{fai:noms4}, and the third requirement is according to the fact that $S''$ should be closed in $p''[]_2\delta_e$.  

In addtion, according to Corollary~\ref{fai:nomt}, $\delta''_p | S'' \subseteq \delta_p - mt$. Finally, according to Assumption \ref{assumption:infinite}, all ocmputations of $p[]\delta_e$ that start in $S$ are infinite. Hence, by condition $C1$, all computations of $\delta''_p []_2 \delta_e$ that start from a state in $S'$ must be infinite.
Our algorithm declares that no solution for the addition problem exists only when there is no subset of $S$ satisfying three requirements above such there all computation of $(\delta_p-mt) []_2 \delta_e$ within that subset are infinite.

\end{proof}

\begin{theorem}
\label{thm:fsP}
Algorithm \ref{ALG:FAILSAFE} is polynomial (in the state space of $p$)
\end{theorem}

\begin{proof}
The proof follows from the fact that each statement in Algorithm \ref{ALG:FAILSAFE} is executed in polynomial time and the number of iterations are also polynomial. 
\end{proof}

Proof of Theorem~\ref{thm:fsMainTheorem} is resulted from Theorem~\ref{fail:soundness}, \ref{fail:completeness}, and \ref{thm:fsP}.

\section{Proofs for Algorithm \ref{ALG:MASKING}}
\label{sec:maskingproofs}

In this section, we prove the soundness,  completeness, and complexity result of Algorithm \ref{ALG:MASKING}.

\begin{lemma}
\label{masking:noMs1}
In all computations $ \langle s_0, s_1, \ldots \rangle$ of $p'[]_2\delta_e[]f$ where $s_0 \in S'$, there dose not exist $s_i$ such that $s_i$ is in $ms_1$ in some iteration of loop on Lines~\ref{mas:outerLoopBegin}-\ref{mas:outerLoopEnd}. 
\end{lemma}

\begin{proof}
Consider a computation $ \langle s_0, s_1, \ldots \rangle$ of $p'[]_2\delta_e$ where $s_0 \in S'$, and there exists $s_i$ such that $s_i$ is in $ms_1$ some iteration of loop on Lines~\ref{mas:outerLoopBegin}-\ref{mas:outerLoopEnd}. It follows that $(s_{i-1} ,s_i)$ in in $mt$.
By construction, $p'$ does not contain any transition in $mt$. Thus, $(s_{i-1}, s_n)$ is a transition of $f$ or $\delta_e$. If it is in $f$ then $s_{i-1} \in ms_1$ (i.e., $s_{i-1} \in ms_2$). If it is in $\delta_e$ then $s_{i-1} \in ms_2$. Therefore, in both cases, $s_{n-1} \in ms_2$, and $(s_{n-2}, s_{n-1}) \in mt$. Again, by construction we know that $\delta'_p$ does not contain any transition in $mt$, so $(s_{n-2}, s_{n-1})$ is either in $f$ or $\delta_e$. If it is in $f$ then $s_{n-2} \in ms_1$ (i.e., $s_{n-2} \in ms_2$). If it is in $\delta_e$ two cases are possible:

1) $(s_{n-1}, s_n) \in f$. In this case, as stated before, $s_{n-1} \in ms_1$, so $s_{n-2} \in ms_2$. 

2) $(s_{n-1}, s_n) \in \delta_e$. In this case, as both $(s_{i-2},s_{i-1})$ and $(s_{i-1}, s_i)$ are in $\delta_e$, according to fairness assumption, there does not exist a transition $\delta'_p - mt$ starting from $s_{i-1}$, and it means that $s_{i-1}$ is added to  $ms_1$  by line~\ref{masking:addtoMS1}, so $s_{i-2} \in ms_2$.\footnote{Note that, as $s_{n-1} \in ms_2$ it is not in $S'$. Thus, no transition of $\delta'_p$ starting from $s_{i-1}$ is removed in $RemoveDeadlock$ or $EnsureClosure$ functions.}

Continuing this argument further leads to the conclusion that $s_0 \in ms_2$. This is a contradiction. Thus, In all computations $ \langle s_0, s_1, \ldots \rangle$ of $p'[]_2\delta_e$ where $s_0 \in S'$, there dose not exist $s_i$ such that $s_i$ is in $ms_1$ some iteration of loop on Lines~\ref{mas:outerLoopBegin}-\ref{mas:outerLoopEnd}.

\end{proof}

\begin{lemma}
\label{masking:RpByE}
In every computation $ \langle s_0, s_1, \ldots \rangle$ of $p'[]_k \delta_e[] f$ that starts from a state in $S'$, if there exists state $s_i$ in $R_p$, then $(s_{i-1}, s_i) \in \delta_e$.
\end{lemma}
\begin{proof}
Every state in $R_p$ is in $ms_2$, and every transition to a state in $ms_2$ is in $mt$ in some iteration of loop on Lines~\ref{mas:outerLoopBegin}-\ref{mas:outerLoopEnd}. By construction, $p'$ does not contain any transition which is in $mt$ in some iteration of loop on Lines~\ref{mas:outerLoopBegin}-\ref{mas:outerLoopEnd}. Thus, $(s_{i-1}, s_i) \notin \delta'_p$.

In addition  $(s_{i-1}, s_i)$ cannot be in $f$, because if  $(s_{i-1}, s_i) \in f$ then $s_{i-1} \in ms_1$ which according to Lemma~\ref{masking:noMs1} is impossible.
\end{proof}

\begin{lemma}
\label{masking:convergence}
For every computation $\langle s_0, s_1, \ldots \rangle$ of $\delta'_p[]_k\delta_e []f$ that starts from a state in $S'$ we have:
$(\exists i : i \geq 0: s_i \in (R \cup R_p) - S') \Rightarrow (\exists j : j > i : s_j \in S')$.
\end{lemma}
\begin{proof}
There are two cases:\\ \\ 
\textbf{Case 1} $s_i \in R$ \\
As $s_i \in R$, according to Lemma~\ref{self:RtoS}, $\exists j : j > i : s_j \in S'$. \\ \\
\textbf{Case 2}$s_i \in R_p$ \\ 
As $s_i$ is in $R_p$, there is a program transition from $s_i$ to a state $s$ in $R$. As $s_0 \in S'$, according to Lemma~\ref{masking:RpByE}, $(s_{i-1}, s_i) \in \delta_e$, and because of fairness assumption program can reach $R$ using $(s_i, s)$.
\end{proof}

\begin{lemma}
\label{masking:f-span}
$R \cup R_p$ is a $f$-span for $p'[]_2\delta_e$ from $S'$.
\end{lemma}
\begin{proof}
By construction, we know that $S' \subseteq R$, thus $S' \Rightarrow (R \cup R_p)$. Any state in $\neg (R \cup R_p)$ is in $ms_1$ in some iteration of loop on Lines~\ref{mas:outerLoopBegin}-\ref{mas:outerLoopEnd}. According to Lemma~\ref{masking:noMs1}, there is no  computations of $p'[]_2\delta_e[]f$ where $s_0 \in S'$ such that $s_i$ is in $ms_1$ in some iteration of loop on Lines~\ref{mas:outerLoopBegin}-\ref{mas:outerLoopEnd}. Therefore, $R \cup R_p$ is a $f$-span for $p'[]_2\delta_e$.

\end{proof}

\begin{theorem}
\label{mas:soundness}
Algorithm~\ref{ALG:MASKING} is sound.
\end{theorem}

\begin{proof}
In order to show the soundness of our algorithm, we need to show that the three conditions of the problem statement are satisfied.

\textit{\textbf{C1}}: Satisfaction of \textit{C1} for Algorithm~\ref{ALG:MASKING} is the same as that for Algorithm~\ref{ALG:FAILSAFE} stated in the proof of the Theorem~\ref{fail:soundness}.

\textit{\textbf{C2}}: From \textit{C1} and the assumption that $p[]_2\delta_e$ refines $spec$ from $S$, $p'[]_2\delta_e$ refines $spec$ from $S'$.

Let $spec = \langle Sf, Lv \rangle$. Consider prefix $c$ of $p'[]_k \delta_e []f$  such that $c$ starts from a state in $S'$. If $c$ does not refine $Sf$, there exists a prefix of $c$, say $\langle s_0, s_1,\ldots, s_n \rangle$, such that it has a transition in $\delta_b $. Wlog, let $\langle s_0, s_1, \cdots, s_n\rangle$ be the smallest such prefix. It follows that $(s_{n-1}, s_n) \in \delta_b$. Hence, $(s_{n-1}, s_n) \in mt$. By construction, $p'$ does not contain any transition in $mt$. Thus, $(s_{n-1}, s_n)$ is a transition of $f$ or $\delta_e$. If it is in $f$ then $s_{n-1} \in ms_1$ which it is a contradiction to Lemma~\ref{masking:noMs1}. If it is in $\delta_e$ then $s_{n-1} \in ms_2$, and $(s_{n-2}, s_{n-1}) \in mt$. Again, by construction we know that $\delta'_p$ does not contain any transition in $mt$, so $(s_{n-2}, s_{n-1})$ is either in $f$ or $\delta_e$. If it is in $f$ then $s_{n-2} \in ms_1$ (contradiction to  Lemma~\ref{masking:noMs1}). If it is in $\delta_e$, as both $(s_{n-2},s_{n-1})$ and $(s_{n-1}, s_n)$ are in $\delta_e$, according to fairness assumption, there should does not exist a transition of $\delta'_p - mt$ starting from $s_{n-1}$, and it means that $s_{n-1} \in ms_1$, which is again a contradiction to  Lemma~\ref{masking:noMs1}. Thus, each prefix of $c$ does not have a transition in $\delta_b$. Therefore, any prefix of $p'[]_k \delta_e []f$ refines $Sf$.

As $p'[]_2\delta_e$ refines $spec$ from $S'$, any prefix of $p'[]_k \delta_e []f$ refines $Sf$, and according to  Lemma~\ref{masking:convergence} and Lemma~\ref{masking:f-span}, $p'$ is masking 2-$f$-tolerant to $spec$ from $S'$ in environment $\delta_e$.

\textit{\textbf{C3}}: Any $(s_0, s_1) \in \delta_r$, is in $mt$. By construction, $p'$ does not have any transition in $mt$, so \textit{C3} holds.

\end{proof}

\begin{observation}
\label{mas:ms1}
In each iteration of loop on Lines \ref{mas:outerLoopBegin}-\ref{mas:outerLoopEnd}, there are two cases for any state $s_0 \in ms_1$: 
\begin{enumerate}
\item $s_0 \in \neg (R \cup R_p)$
\item $s_0$ is added by Lines~\ref{mas:expandingMsBeign}-\ref{mas:expandingMsEnd}
\end{enumerate}
\end{observation}

\begin{theorem}
\label{mas:noMs2}
Let $p''$ be any program that solves the problem of adding masking fault-tolerance and let $S''$ be its invariant. Then, $S''$ does not include any state in the set $ms_2$ in any iteration of loop on Line~\ref{mas:outerLoopBegin}-\ref{mas:outerLoopEnd}.
\end{theorem}
\begin{proof}
The proof of this theorem is based on Observation~\ref{mas:ms1}, extension of  Theorem~\ref{fai:reachms1}  and Theorem~\ref{fai:ms1} from failsafe fault-tolerance to masking fault-tolerance, and extension of  Corollary~\ref{notRNotRecovery} from stabilizaton to masking fault-tolerance.

\end{proof}

\begin{corollary}
\label{mas:noMt}
Let $p''$ be any program that solves the problem of adding masking fault-tolerance and let $S''$ be its invariant. $p''|S''$ cannot have any transition in $mt$ in any iteration of loop on Lines~\ref{mas:outerLoopBegin}-\ref{mas:outerLoopEnd}. 
\end{corollary}

\begin{theorem}
\label{mas:noMs4}
Let $p''$ be any program that solves the problem of adding masking fault-tolerance, and let $S''$ be its invariant. Then $S''$ cannot include any state in set $ms_4$ in any iteration of loop on Lines~\ref{mas:ms4Begin}-\ref{mas:ms4End}. 
\end{theorem}
\begin{proof}
The proof of this theorem is the same as that of Theorem~\ref{fai:noms4}.
\end{proof}

\begin{theorem}
\label{mas:completeness}
Algorithm~\ref{ALG:MASKING} is complete.
\end{theorem}
\begin{proof}
Let program $p''$ and predicate $S''$ solve transformation problem. $S''$ should satisfy following requirements: 
\begin{enumerate}
\item $S''$ does not include any state in set $ms_2$ in any iteration of loop on Lines~\ref{mas:outerLoopBegin}-\ref{mas:outerLoopEnd}
\item $S''$ does not include any state in set $ms_4$ in any iteration of loop on Lines~\ref{mas:ms4Begin}-\ref{mas:ms4End}
\item $\nexists s_0 : s_0 \in S'' : (\exists s_1 : s_1 \notin S'' : (s_0, s_1) \in \delta_e)$
\end{enumerate}

The first requirement is according to Theorem~\ref{mas:noMs2}. The second requirement is according to Theorem~\ref{mas:noMs4}, and the third requirement is according to the fact that $S''$ should be closed in $p''[]_2\delta_e$.  

In addtion, according to Corollary~\ref{mas:noMt}, $\delta''_p | S''$ cannot have any transition in $mt$ in any iteration of loop on Lines~\ref{mas:outerLoopBegin}-\ref{mas:outerLoopEnd}. Finally, according to Assumption \ref{assumption:infinite}, all ocmputations of $p[]\delta_e$ that start in $S$ are infinite. Hence, by condition $C1$, all computations of $\delta''_p []_2 \delta_e$ that start from a state in $S'$ must be infinite.
Our algorithm declares that no solution for the addition problem exists only when there is no subset of $S$ satisfying three requirements above such there all computation of $(\delta_p-mt) []_2 \delta_e$ within that subset are infinite.

\end{proof}

\begin{theorem}
\label{thm:mP}
Algorithm \ref{ALG:MASKING} is polynomial (in the state space of $p$)
\end{theorem}

\begin{proof}
The proof follows from the fact that each statement in Algorithm \ref{ALG:MASKING} is executed in polynomial time and the number of iterations are also polynomial. 
\end{proof}

Proof of Theorem~\ref{thm:mMainTheorem} is resulted from Theorem~\ref{mas:soundness}, \ref{mas:completeness}, and \ref{thm:mP}.

\end{document}